%% file: main.tex
\documentclass[letterpaper,conference]{ieeeconf}
\IEEEoverridecommandlockouts                              
\usepackage[utf8]{inputenc}
\usepackage{textcomp}

\usepackage{graphicx}
\usepackage{color}
\usepackage[table]{xcolor}
\usepackage{hhline}
\usepackage{amsmath,amsthm}
\usepackage{mathtools}
\usepackage[version=4]{mhchem}
\usepackage{siunitx}
\usepackage{graphicx}
\usepackage{longtable,tabularx}
\usepackage{booktabs}
\usepackage{multirow}
\usepackage{hhline}
\usepackage{dsfont}
\usepackage[linesnumbered,ruled,vlined]{algorithm2e}
\SetKwInOut{KwParam}{Parameters}
\let\oldnl\nl
\newcommand{\nonl}{\renewcommand{\nl}{\let\nl\oldnl}}
\DontPrintSemicolon
\usepackage[style=ieee ,sorting=none,maxbibnames=99,giveninits, doi=false, backend=bibtex]{biblatex}
\AtBeginBibliography{\scriptsize}

\input{PackagesCommands}
\newcommand{\argmin}{\mathop{\mathrm{argmin}}\limits}

\input{JUcmds}

\title{Computing Safe Control Inputs using Discrete-Time Matrix Control Barrier Functions via Convex Optimization}

\author{
James Usevitch, 
Juan Augusto Paredes Salazar,
and
Ankit Goel
\thanks{James Usevitch is with the Department of Electrical and Computer Engineering, Brigham Young University, Provo, UT 84602. {\tt\small james\_usevitch@byu.edu}}%
\thanks{Juan Augusto Paredes Salazar and Ankit Goel are with the Department of Mechanical Engineering, University of Maryland, Baltimore County,1000 Hilltop Circle, Baltimore, MD 21250. {\tt\small japarede, ankgoel@umbc.edu }}%
}

\begin{document}

\maketitle

\begin{abstract}
Control barrier functions (CBFs) have seen widespread success in providing forward invariance and safety guarantees for dynamical control systems.
%
The majority of formulations consider continuous-time formulations, but recent focus has shifted to discrete-time dynamics.
A crucial limitation of discrete-time formulations is that CBFs that are nonconcave in their argument require the solution of nonconvex optimization problems to compute safety-preserving control inputs, which inhibits real-time computation of control inputs guaranteeing forward invariance.
This paper presents a novel method for computing safety-preserving control inputs for discrete-time systems with nonconvex safety sets, utilizing convex optimization and the recently developed class of matrix control barrier function techniques. 
The efficacy of our methods is demonstrated through numerical simulations on a bicopter system.
\end{abstract}
%

\section{Introduction}

Guaranteeing safety and collision avoidance for dynamical systems is a fundamental problem in control theory and robotics.
Control barrier functions (CBFs) have become a widely applied technique to obtain safety guarantees for systems with applications in aerospace \cite{breeden2023autonomous, molnar2025collision}, robotics \cite{ferraguti2022safety, cortez2019control}, self-driving cars \cite{chen2017obstacle, alan2023control}, and numerous other domains \cite{garg2024advances}. 
The majority of prior work on CBFs has considered continuous-time systems, but a growing body of work has analyzed CBFs from a discrete-time perspective \cite{agrawal2017discrete, Cosner-RSS-23, zeng2021safety, liu2023iterative, khajenejad2021tractable, zeng2021safety}.

One key challenge of discrete-time CBFs is that the optimization formulation for computing a safety-preserving control input is generally nonconvex.
Unlike continuous-time CBFs, the safety constraint is, in general, non-affine in the control input; therefore, the convexity of the optimization problem heavily depends on the form of the function defining the safe set. 
Many common collision avoidance scenarios require the use of nonconvex safe set functions, which necessitate the use of nonconvex optimization solvers that have longer runtimes and give few, if any, guarantees on finding a globally optimal solution.

Prior work has dealt with this challenge in several ways. In \cite{agrawal2017discrete}, the authors considered linear systems and convex quadratic discrete-time CBFs to maintain convexity of the optimization program computing safe control inputs. 
The authors of \cite{ahmadi2019safe, zeng2021safety, Cosner-RSS-23} assume that the composition of the CBF with the dynamics is concave in the control input. This results in a convex optimization problem; however, the assumption does not hold for many common collision avoidance scenarios involving convex obstacles.
In \cite{khajenejad2021tractable}, the notion of partially affine CBFs was introduced, which guarantees convexity of the optimization problem under a specific form of partitioned control-affine discrete-time dynamics.
Despite the progress made by these prior works, the challenge of computing safe control inputs for discrete-time systems using solely convex optimization remains an open problem.

\begin{figure}
    \centering
    \includegraphics[width=0.9\columnwidth]{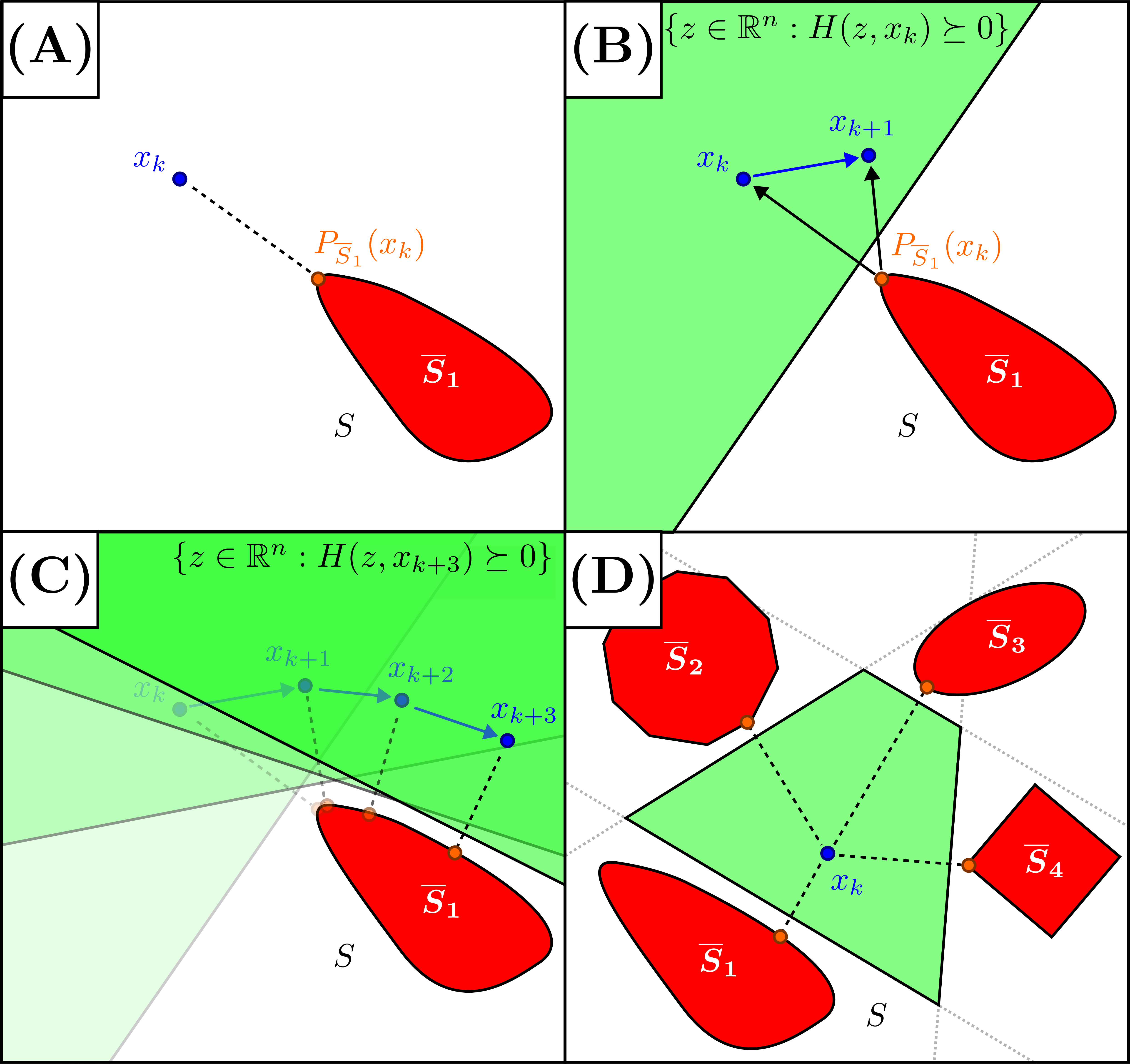}
    \caption{Illustration of our proposed method to maintain convexity of computing safe control inputs for discrete-time systems. \textbf{(A):} The current state is projected onto the unsafe set. \textbf{(B):} A halfspace is formed which excludes the obstacle, and a control input keeping the next state in the halfspace is computed. \textbf{(C):} The process is repeated iteratively. \textbf{(D):} The proposed method can be applied to scenarios where the unsafe set is a subset of a union of convex regions.}
    \label{fig:illustration}
\end{figure}

In addition, prior work has only considered scalar-valued control barrier functions.
Recently, the notion of continuous-time matrix control barrier functions (MCBFs) was introduced \cite{ong2025matrix, ong2025combinatorial}. 
MCBFs, of which scalar-valued CBFs are a special case, allow for more general classes of constraints over the cone of positive semidefinite matrices using the Loewner order.
To our knowledge, the notion of MCBFs has not been extended to discrete-time dynamics.

This paper introduces the concept of exponential discrete-time matrix control barrier functions and presents a novel method for computing safe control inputs using solely convex optimization tools. 
Our proposed method operates by iteratively constructing convex subsets of the safe set obtained through a projection-based method.
The method only requires solving two sequential convex optimization problems, and our experimental results suggest that the method is considerably faster and more reliable than nonconvex optimization approaches.
The main contributions of this work thus are 
\begin{itemize}
    \item A novel definition of exponential discrete-time matrix control barrier functions,
    \item A novel method to compute safe control inputs for discrete-time systems with nonconvex safe sets using solely convex optimization. This includes the new notions of \emph{safe subset functions} and \emph{Subset-based Discrete-Time Exponential Matrix Control Barrier Functions} (SDTE-MCBFs).
    \item A novel extension of indefinite matrix safe sets to discrete-time matrix control barrier functions
\end{itemize}
We note that the work \cite{liu2023iterative} considers a method of projecting onto convex obstacles, similar to the one presented in this paper. 
However, this prior work only considers the special case of scalar-valued safe set functions, does not consider the zeroing CBF property of asymptotic convergence to the safe set, and does not consider disjunctive boolean compositions of multiple CBFs.

The outline of this paper is as follows: notation and the problem formulation are given in Section \ref{sec:notation_problem}. Our main results are given in Section \ref{sec:main_results}. Numerical simulations demonstrating the efficacy of our proposed method are given in \ref{sec:simulations}. A brief conclusion is given in Section \ref{sec:conclusions}.

\section{Notation and Problem Formulation}
\label{sec:notation_problem}

The closure of a set \(S\) is denoted \(\text{cl}(S)\), the convex hull is denoted \(\text{co}(S)\), and the power set is denoted \(2^S\).
The set of symmetric real matrices of size \(p\times p\) is denoted \(\mathbb{S}^p\). The cone of positive semidefinite symmetric \(p \times p\) matrices is denoted \(\mathbb{S}_+^p\), and the set of positive definite \(p \times p\) symmetric matrices is denoted \(\mathbb{S}_{++}^p\).
Given matrices \(A,B \in \R^{n\times n}\), the notation \(A \ggeq 0\) and \(B \succ 0\) indicate that \(A\) is positive semidefinite and \(B\) is positive definite as per the Loewner order.

The following theorem will be used to bound the eigenvalues of the matrices in \(\mathbb{S}^p\).
\begin{theorem}[{Weyl's Inequality \cite[Thm 4.3.1]{horn2012matrix}}]
    \label{thm:weyl}
    Let \(A, B \in \mathbb{S}^p\). Let the respective eigenvalues of \(A,B\) and \(A+B\) be \(\{\lambda_i(A)\}_{i=1}^p\), \(\{\lambda_i(B)\}_{i=1}^p\), and \(\{\lambda_i(A+B)\}_{i=1}^p\). Let the eigenvalues be ordered such that \(\lambda_1 \leq \lambda_2 \leq \cdots \leq \lambda_p\). Then the following holds:
    \begin{align}
        \lambda_{i-j+1}(A) + \lambda_{j}(B) \leq \lambda_i(A+B).
    \end{align}
\end{theorem}

The following will be required for our analysis of convex sets. 
\begin{define}[Projection Mapping]
    Given a closed, convex set \(C \subset \R^n\), the projection mapping \(P_C : \R^n \to \R^n \) is defined as \(P_C(x) = \arg\min_{z\in C} \nrm{z-x}\).
\end{define}

\begin{define}[{Normal Cone \cite{rockafellar1998variational}}]
    \label{def:normal_cone}
    The normal cone mapping of a convex set \(C \subset \R^n\), denoted \(N_C : \R^n \to 2^{\R^n}\) is defined as \(N_C(x) \triangleq \{v \in \R^n : v^\intercal (z-x) \leq 0 \text{ for all } z \in C \}.\)
\end{define}

\begin{theorem}[{\cite[Prop. 6.17]{rockafellar1998variational}}]
    \label{thm:Normal_Projection_Cones}
    Given a convex set \(C \subset \R^n\), the normal cone mapping \(N_C\) and projection mapping \(P_C\) are related as follows, where \(I\) indicates the identity mapping:
    \begin{align}
        N_C &= P_C^{-1} - I, &
        P_C &= (I + N_C)^{-1}. \label{eq:P_C_N_C}
    \end{align}
\end{theorem}

\subsection{Background on Discrete-Time Safety Filters}

Consider a discrete-time dynamical system with the control-affine form
\begin{align}
    x_{k+1} = f(x_k) + g(x_k) u_k. \label{eq:discrete_dynamics}
\end{align}
Here, \(x_k \in \R^n\) denotes the state at step \(k\), \(u \in \R^m\) denotes the control input at step \(k\), and the functions \(f : \R^n \to \R^n \) and \(g : \R^n \to \R^{n \times m}\) are locally Lipschitz in their arguments. Without loss of generality, we assume an initial step value of \(k = 0\).

The state space \(\R^n\) is divided into a safe set \(S_h \subset \R^n\) and an unsafe set \(\overline{S}_h = \R^n \backslash S_h\). The set \(S_h\) is defined as the superlevel sets of a function \(h : \R^n \to \R\) as follows:
\begin{subequations}
\label{eq:scalar_h_function}
\begin{align}
    S_h &= \{x \in \R^n : h(x) \geq 0\}, \\
    \text{int}(S_h) &= \{x \in \R^n : h(x) > 0\}, \\
    \partial S_h &= \{x \in \R^n : h(x) = 0\}.
\end{align}
\end{subequations}
The objective is to render the set \(S_h\) forward invariant with respect to the dynamics \eqref{eq:discrete_dynamics}, which is defined as follows:
\begin{define}
    A set \(\mathcal{C} \subset \R^n\) is forward invariant with respect to the dynamics \eqref{eq:discrete_dynamics} if \(x_k \in \mathcal{C}\) for all \(k \geq 0\).
    The set \(\mathcal{C}\) is called \emph{safe} if it is both forward invariant and \(\mathcal{C} \subseteq S_h\).
\end{define}
Let \(\mathcal{D}_h \subseteq \R^n\) be a domain such that \(S_h \subset \text{int}(\mathcal{D}_h)\).
It has been shown in prior literature that satisfaction of the following condition for all \(x_k \in \mathcal{D}_h\) is sufficient to ensure forward invariance of the set \(S_h\) \cite{agrawal2017discrete}:
\begin{align}
    &\sup_{u} \bkt{h(x_{k+1}(x_k,u)) - h(x_k)} \geq - \gamma h(x_k), \label{eq:safety_condition} \\
    &\iff \sup_u \bkt{h(f(x_k) + g(x_k) u)} \geq h(x_k) - \gamma h(x), \nonumber 
\end{align}
where \(0 < \gamma \leq 1\).

\begin{define}
    A function \(h : \R^n \to \R\) is called a discrete-time exponential control barrier function for the system \eqref{eq:discrete_dynamics} if the condition \eqref{eq:safety_condition} holds for all \(x \in \mathcal{D}_h\).
\end{define}

We define the set 
\begin{align}
&K_h(x_k) =  \nonumber \\
&\{u \in \R^m : h(f(x_k) + g(x_k) u) \geq h(x_k) - \alpha(h(x)) \}. \label{eq:Kh_exp_CBF}
\end{align}
\begin{theorem}[\cite{agrawal2017discrete}]
    Suppose that \eqref{eq:safety_condition} holds for all \(x_k \in \mathcal{D}_h\) and \(h(x_0) \geq 0\). Then any control law \(u_k(\cdot)\) such that \(u_k(x_k) \in K_h(x_k)\) renders the set \(S_h\) forward invariant.
\end{theorem}
In words, \(K_h\) is the set of control inputs that render the set \(S_h\) forward invariant. A control input \(u_k \in K_h(x_k)\) is called \emph{safe}. It is also straightforward to show that control inputs \(K_h\) also render the set \(S\) exponentially stable in \(\mathcal{D}_h\).

Suppose a nominal feedback control law \(u_\text{nom}(\cdot)\) is given for system \eqref{eq:discrete_dynamics} above.
Suppose that several discrete-time CBFs \(\{h_i(x_k)\}_{i=1,\ldots,M}\) are given.
If the condition \eqref{eq:safety_condition} holds for all \(h_i\), a safe control input \(u^*(x_k) \in K_h(x_k)\) that minimally modifies \(u_\text{nom}\) in the sense of the Euclidean norm can be computed via the following (possibly nonconvex) optimization problem:
\begin{subequations}
\label{eq:multi_discrete_CBF_opt}
\begin{align}
    &u^*(x_k) = 
    \underset{u \in \R^m}{\arg\min} \nrm{u - u_\text{nom}(x_k)}_2 \\
    &\hspace{1em} \text{s.t.} \hspace{0.5em} h_i(f(x_k) + g(x_k) u) \geq h_i(x_k) -  \gamma_i h_i(x_k)  \label{eq:nonconvex_constraint} \\
    &\hspace{3em} \forall i \in \{1,\ldots,M\}.
\end{align}
\end{subequations}

Several limitations exist to this notion of discrete-time control barrier function and safety filtering as presented above. First, observe that the constraints in \eqref{eq:nonconvex_constraint} are convex if and only if each \(h_i(\cdot)\) is concave in its argument.\footnote{Recall that the composition of a function concave (convex) in its argument \(h(\cdot)\) with an affine mapping \(u \mapsto f(x_k) + g(x_k)u\) is likewise concave (convex) \cite[Sec. 3.2.2]{boyd2004convex}.}
If any \(h_i(\cdot)\) is nonconcave, then the entire optimization problem \eqref{eq:multi_discrete_CBF_opt} is nonconvex. This can inhibit real-time computation of the control input \(u^*(x_k)\) on control systems. 

In addition, the formulation \eqref{eq:multi_discrete_CBF_opt} is unable to represent more complex semidefinite matrix constraints, including constraints on matrix eigenvalues and safe sets in the form of spectrahedra. The recent work \cite{ong2025matrix} introduces the notion of \emph{matrix control barrier functions} to address this gap.
More specifically, given a matrix valued function \(H : \R^n \to \mathbb{S}^p\), they consider safe sets of the form
\begin{align}
    S = \{x \in \R^n : H(x) \ggeq 0 \}. \label{eq:matrix_set_constraint}
\end{align}
However, \cite{ong2025matrix} considers only continuous-time dynamical systems.
Conditions to render sets of the form \(S\) forward invariant for discrete-time systems have not been considered in the literature.

This paper presents novel methods to address both of these gaps. Our methods are applicable to matrix-valued safe sets of the form \eqref{eq:matrix_set_constraint}, of which the scalar-valued safe sets of the form \eqref{eq:scalar_h_function} are a special case. The problem addressed by this paper is as follows:
\begin{problem}
    Given the discrete-time dynamics \eqref{eq:discrete_dynamics}, derive a method to compute control inputs \(u_k\) rendering sets of the form in \eqref{eq:matrix_set_constraint} forward invariant using solely convex optimization methods.
\end{problem}

\section{Main Results}
\label{sec:main_results}

\subsection{Discrete-Time Exponential Matrix Control Barrier Functions}

We consider safe sets \(S\) of the form \eqref{eq:matrix_set_constraint}. The following result presents a discrete-time condition under which the set \(S\) is rendered forward invariant.
\begin{lemma}
    Consider the dynamics \eqref{eq:discrete_dynamics} and the set \(S\) in \eqref{eq:matrix_set_constraint}. Let \(\mathcal{D} \supset S\) be an open superset of \(S\). The set \(S\) is forward invariant if there exists a constant \(0 < \gamma \leq 1\) such that the following condition holds for all \(x \in \mathcal{D}\):
    \begin{align}
        H(x_{k+1}) - H(x_k) \ggeq -\gamma H(x_k). \label{eq:matrix_safety_long}
    \end{align}
\end{lemma}

\begin{proof}
    We proceed by induction.
    By the definition of forward invariance, assume that \(x_0 \in S\). This implies that \(H(x_0) \ggeq 0\).
    Observe that \eqref{eq:matrix_safety_long} can be rewritten as
    \begin{align}
        H(x_{k+1}) \ggeq (1-\gamma)H(x_k). \label{eq:matrix_safety_short}
    \end{align}
    Since \(0 < \gamma \leq 1\), the matrix \((1-\gamma)H(x_0)\) is positive semidefinite. Equation \eqref{eq:matrix_safety_short} implies that \(H(x_1)\) is therefore positive semidefinite.
    This can be seen by noting from \eqref{eq:matrix_safety_short} that \(H(x_{1}) - (1-\gamma)H(x_0) \ggeq 0 \implies y^\intercal (H(x_{1}) - (1-\gamma)H(x_0)) y \geq 0\) for all \(y \in \R^p\), implying that \(y^\intercal H(x_1) y \geq y^T (1-\gamma) H(x_0) y \geq 0\) since \((1-\gamma)H(x_0)\) is positive semidefinite. It follows that \(H(x_1) \ggeq 0\).
    Now assume that \(H(x_k)\) is positive semidefinite for any \(k \geq 0\). 
    This implies that \((1-\gamma)H(x_k)\) is also positive semidefinite.
    Similar arguments can be used to conclude that \(H(x_{k+1}) \ggeq 0\), which concludes the proof.
\end{proof}

\begin{define}
    \label{def:DT_MCBF}
    A function \(H : \mathbb{R}^n\to \mathbb{S}^p\) is called an discrete-time exponential matrix control barrier function (DTE-MCBF) for the system \eqref{eq:discrete_dynamics} if there exists a constant \(0 < \gamma \leq 1\) and a domain \(\mathcal{D} \supset S\) such that for all \(x \in \mathcal{D}\) there exists a \(u \in \R^m\) satisfying
    \begin{align}
        H(f(x) + g(x)u) \ggeq (1-\gamma) H(x)
    \end{align}
\end{define}

We define the set of safety-preserving control inputs \(K_H : \R^n \to 2^{\R^m}\) as
\begin{align}
    K_H(x) \triangleq \{u \in \R^m : H(f(x) + g(x)u) \ggeq (1-\gamma) H(x) \}. \label{eq:orig_K_H}
\end{align}

Let the eigenvalues of \(H(x_k)\) be ordered as \(\lambda_1(x_k) \leq \lambda_2(x_k) \leq \cdots \leq \lambda_p(x_k)\).
The following theorem demonstrates that DTE-MCBFs are zeroing CBFs in the sense that all negative eigenvalues converge exponentially to zero under under any control law \(u(x_k) \in K_H(x_k)\).

\begin{theorem}
    \label{thm:zeroing_thm}
    Let \(H\) be a DTE-MCBF for the system \eqref{eq:discrete_dynamics} with constant \(0 < \gamma \leq 1\). For brevity, we denote \(H_k \triangleq H(x_k)\). Let the eigenvalues of \(H_k\) be ordered as \(\lambda_1(H_k) \leq \lambda_2(H_k) \leq \cdots \leq \lambda_p(H_k)\). 
    Suppose that \(u_k \in K_H(x_k)\) for all \(k \geq 0\).
    Then the following holds:
    \begin{align}
        \lambda_1(H_k) \geq (1-\gamma)^k \lambda_1(H_0).
    \end{align}
\end{theorem}

\begin{proof}
    Define \(\Delta_\gamma H_k \triangleq H_{k+1} - (1-\gamma)H_k = H(f(x_k) + g(x_k)u_k) - (1-\gamma)H(x_k)\). By Definition \ref{def:DT_MCBF}, the control input \(u_k \in K_H(x_k)\ \forall k \geq 0\) implies that \(\Delta_\gamma H \ggeq 0\).

    Observe that \(H_{k+1} = (1-\gamma)H_k + \Delta_\gamma H_k\). By Weyl's Inequality (Theorem \ref{thm:weyl}), it follows that
    \begin{align}
        \lambda_1((1-\gamma)H_k) + \lambda_1(\Delta_\gamma H_k) &\leq \lambda_1\pth{(1-\gamma)H_k + \Delta_\gamma H_k}, \nonumber \\
        &= \lambda_1(H_{k+1}). \label{eq:lambda_Hk}
    \end{align}
    However, note that \(\lambda_1(\Delta_\gamma H_k) \geq 0\) since \(\Delta_\gamma H \ggeq 0\) by the definition of \(K_H\) in \eqref{eq:orig_K_H}. This implies that \(\lambda_1((1-\gamma)H_k) \leq \lambda_1((1-\gamma)H_k) + \lambda_1(\Delta_\gamma H_k)\). It follows from \eqref{eq:lambda_Hk} that
    \begin{align}
        \lambda_1(H_{k+1}) &\geq \lambda_1((1-\gamma)H_k) \\
        &= (1-\gamma) \lambda_1(H_k)\ \forall k \geq 0. \label{eq:exponential_scaling}
    \end{align}
    The result follows by repeated application of \eqref{eq:exponential_scaling}:
    \begin{align*}
        \lambda_1(H_k) \geq (1-\gamma) \lambda_1(H_{k-1}) \geq \cdots \geq (1-\gamma)^k \lambda_1(H_0).
    \end{align*}
\end{proof}

In particular, Theorem \ref{thm:zeroing_thm} implies that if \(x_0 \not\in S\), the state \(x_k\) will converge exponentially to the safe set \(S\) under a control law \(u(x_k) \in K_H(x_k)\).

If \(H\) is a DTE-MCBF and a nominal feedback control law \(u_{\rm nom}(x_k)\) is given, a safe control input \(u^*(x_k) \in K_H(x_k)\) can be computed via the following optimization problem:
\begin{subequations}
\label{eq:direct_MCBF_opt}
\begin{align}
        &u^*(x) = \underset{u \in \R^m}{\arg\min} \nrm{u - u_{\rm nom}(x)}_2 \\
        &\hspace{1em} \text{s.t.} \hspace{1em} H(f(x) + g(x)u) \ggeq (1-\gamma) H(x) \label{eq:MCBF_constraint}
\end{align}
\end{subequations}

However, this optimization problem may not be convex in general due to the constraint \eqref{eq:MCBF_constraint}. In the next section we address this issue.

\subsection{Computing Safe Inputs Using Convex Projection}

We begin by considering the following two classes of safe set functions \(H\). The definition of matrix concavity is given in the Appendix (Section \ref{sec:Appendix}).
\begin{itemize}
    \item Safe sets that are matrix concave in their argument
    \item Safe sets that are not matrix concave in their argument.
\end{itemize}
For safe sets that are matrix concave in their argument, observe that the constraint \eqref{eq:MCBF_constraint} is matrix convex in the optimization variable \(u\). This follows because the composition of a convex function with an affine function is also convex. The optimization problem \eqref{eq:direct_MCBF_opt} is therefore convex and can be solved with standard conic optimization solvers.

We turn our attention to the more difficult case of safe sets \(H\) that are \emph{not} matrix concave in their arguments. Recall that \(H\) defines the safe set \(S\) via its superlevel sets. The following assumption is made on the unsafe set \(\overline{S}\):
\begin{assume}
    \label{assume:unsafe_set}
    Given a safe set \(S\) defined by the sublevel sets of \(H\),
    the closure of the unsafe set \(\text{cl}(\overline{S}) = \text{cl}(\R^n \backslash S) \) is a subset of a finite number of closed, convex sets. More precisely, \(\text{cl}(\overline{S}) \subset \bigcup_{i=1}^Q \overline{S}_i\). In addition, for each \(\overline{S}_i\) there exists a matrix convex function \(\overline{H}_i : \R^n \to \mathbb{S}\) such that \(\overline{S}_i = \{x \in \R^n : \overline{H}_i \gleq 0 \}\).
\end{assume}
In the trivial case any nonconvex unsafe set \(\overline{S}\) is a subset of its convex hull; i.e., \(\text{cl}(\overline{S}) \subset \text{co}(\overline{S})\). As this can be overly conservative, however, Assumption \ref{assume:unsafe_set} also includes cases where the unsafe set can be decomposed into a finite number of closed, convex sets.

Since the set \(S\) is nonconvex and therefore the function \(H\) is not matrix concave, the optimization constraint \eqref{eq:MCBF_constraint} is nonconvex. However, we can instead iteratively consider a \emph{convex subset} of \(S\) based on the current state \(x_k\).
Our proposed method is as follows: At each timestep \(k\) such that \(x_k \not\in \overline{S}_i\), observe that the sets \(\{x_k\}\) and \(\overline{S}_i\) are both nonempty, disjoint, closed convex sets, with the set \(\{x_k\}\) being trivially compact. By the separating hyperplane theorem \cite[Sec. 2.5.1]{boyd2004convex}, there exists a hyperplane defined by \(a_i : \R^n \to \R^n, b_i : \R^n \to \R\) such that \(\pth{a_i(x_k)}^\intercal x_k - b_i(x_k) > 0\) and \(\pth{a_i(x_k)}^\intercal z < b_i(x_k)\) for all \(z \in \overline{S}_i\).

The form of \(a_i, b_i\) is determined by the vector \(x_k - P_{\overline{S}_i}{x_k}\), where \(P_{\overline{S}_i}{x_k}\) is the projection of \(x_k\) onto the set \(\overline{S}_i\).
This projected point can be computed via the following convex optimization problem:
\begin{subequations}
\label{eq:convex_projection}
\begin{alignat}{2}
P_{\overline{S}_i}(x_k) = \arg\min_{z \in \R^n} &&& \nrm{z - x_k} \\
\text{s.t.} &&\hspace{1em}& \overline{H}_i(z) \gleq 0.
\end{alignat}
\end{subequations}
Once this projected point has been computed, the supporting hyperplane parameters \(a_i(x_k), b_i(x_k)\) can be explicitly defined as
\begin{subequations}
\label{eq:a_i_b_i}
\begin{align}
    a_i(x_k) &\triangleq \frac{x_k - P_{\overline{S}_i}(x_k)}{\nrm{x_k - P_{\overline{S}_i}(x_k)}_2}, \\
    b_i^\epsilon (x_k) &\triangleq \epsilon +  \pth{\frac{x_k - P_{\overline{S}_i}(x_k)}{\nrm{x_k - P_{\overline{S}_i}(x_k)}_2}}^\intercal P_{\overline{S}_i}\pth{x_k},
\end{align}
\end{subequations}
for some \(\epsilon > 0\).

\begin{lemma}
    \label{lem:unsafe_inequality}
    Let \(a_i, b_i^\epsilon\) be defined as in \eqref{eq:a_i_b_i} for some \(\epsilon > 0\). Suppose that \(x_k \notin \overline{S}_i\). Then the following holds: 
    \begin{align}
        a_i(x_k)^\intercal z < b_i(x_k)\ \forall z \in \overline{S}_i.
    \end{align}
\end{lemma}
\begin{proof}
    We first demonstrate that the vector \(x_k - P_{\overline{S}_i}(x_k)\) is in the normal cone \(N_{\overline{S}_i}(P_{\overline{S}_i}(x_k))\). Using Theorem \ref{thm:Normal_Projection_Cones}, observe that
    \begin{align}
        x_k - P_{\overline{S}_i}(x_k) &\in P_{\overline{S}_i}^{-1}\pth{P_{\overline{S}_i}(x_k)}- I(P_{\overline{S}_i}(x_k)) \\
        &= N_{\overline{S}_i}\pth{P_{\overline{S}_i}(x_k)}
    \end{align}
    Since \(\nrm{x_k - P_{\overline{S}_i}(x_k)}_2 > 0\) it follows that \( \frac{x_k - P_{\overline{S}_i}(x_k)}{\nrm{x_k - P_{\overline{S}_i}(x_k)}_2} \in N_{\overline{S}_i}(P_{\overline{S}_i}(x_k))\)
    Next, by Definition \ref{def:normal_cone} it holds that \(\pth{ \frac{x_k - P_{\overline{S}_i}(x_k)}{\nrm{x_k - P_{\overline{S}_i}(x_k)}_2}}^\intercal \pth{z - P_{\overline{S}_i}(x_k)} \leq 0 < \epsilon\ \forall z \in \overline{S}_i\).
    Rearranging yields
    \begin{align*}
        &\pth{ \frac{x_k - P_{\overline{S}_i}(x_k)}{\nrm{x_k - P_{\overline{S}_i}(x_k)}_2}}^\intercal z <  \\
        &\hspace{4.5em} \epsilon - \pth{ \frac{x_k - P_{\overline{S}_i}(x_k)}{\nrm{x_k - P_{\overline{S}_i}(x_k)}_2}}^\intercal P_{\overline{S}_i}(x_k)\ \forall z \in \overline{S}_i
    \end{align*}
    which by \eqref{eq:a_i_b_i} is equivalent to \(a_i(x_k)^\intercal z < b_i^\epsilon(x_k)\) for all \(z \in \overline{S}_i\).
\end{proof}
    Lemma \ref{lem:unsafe_inequality} implies that \(\overline{S}_i\) is a subset of the interior of the halfspace defined by \(\{y \in \R^n : a_i(x_k)^\intercal y < b_i^\epsilon (x_k)\}\), \(y \in \R^m\) for all values of \(\epsilon > 0\).
    On the other hand, by definition observe that the condition \(a_i(x_k)^\intercal z - b_i^\epsilon \geq 0\) is equivalent to
    \begin{align}
        \pth{\frac{x - P_{\overline{S}_i}(x_k)}{\nrm{x - P_{\overline{S}_i}(x_k)}}}^\intercal \pth{z - P_{\overline{S}_i}(x_k)} \geq \epsilon.
    \end{align}
    When \(x_k \not\in \overline{S}_i\), the parameter \(\epsilon\) ensures that the equations \eqref{eq:a_i_b_i} are well-defined by controlling the distance between the hyperplane \(a_i(x_k)^\intercal y = b_i^\epsilon(x_k)\) and the boundary of the set \(\overline{S}_i\). 
    %
    A positive \(\epsilon\) adds a buffer between the boundary and the hyperplane;
    thus, the parameter \(\epsilon\) is called the \textit{buffer distance}.
    %


Now consider the following revised safe set functions:
\begin{align}
    h_i^\epsilon(z, x_k) &\triangleq a_i(x_k)^\intercal z - b_i^\epsilon(x_k),\ \forall i=1,\ldots,Q, \label{eq:h_prime} \\
    H^\epsilon(z,x_k) &\triangleq \bmx{
    h_1^\epsilon(z,x_k) & & 0 \\ & \ddots & \\ 0 & & h_Q^\epsilon(z,x_k)
    } \label{eq:H_prime_set}
\end{align}
Each \(h_i^\epsilon\) is linear in \(z\), which implies \(H^\epsilon\) is therefore matrix concave in \(z\). In particular, the constraint \(H^\epsilon(z,x_k) \ggeq 0\) enforces  the inequality \(\min_i\pth{\brc{h_i^\epsilon(z, x_k)}_{i=1}^Q} \geq 0\). 

\begin{lemma}
    \label{lem:convex_subset}
    Let \(x_k \not\in \overline{S}_i\) for all \(i=1,\ldots,Q\). 
    Then the set \(\{z \in \R^n : H^\epsilon(z, x_k) \ggeq 0 \}\) is disjoint from \(\overline{S}_i\) for all \(i=1,\ldots,Q\).
\end{lemma}

\begin{proof}
    The result follows from Lemma \ref{lem:unsafe_inequality} and equations \eqref{eq:h_prime}-\eqref{eq:H_prime_set}.
\end{proof}
By Lemma \ref{lem:convex_subset}, the set \(\{z \in \R^n : H^\epsilon(z, x_k) \ggeq 0 \}\) is a convex subset of the safe set \(S\).
However, observe that each \(h_i^\epsilon\) and \(H^\epsilon\) is a two-argument function where the second argument determines the subset of the safe set being considered and the first argument is the point whose safety is being determined.
This two-argument form motivates a new class of safe set functions and control barrier functions which are outlined in the next two definitions.

\begin{define}
    \label{def:safe_subset_function}
    Let \(S \subset \R^n\).
    A function \(H : \R^n \times \R^n \to \mathbb{S}^p\) is called a \emph{weak safe subset function} for \(S\) if 
    there exists a domain \(\mathcal{D} \supset S\) such that 
    the following holds for all \(x \in \mathcal{D}\):
    \begin{align}
        \{z \in \R^n : H(z,x) \ggeq 0 \} \subseteq S.
    \end{align}
    A function \(H\) is called a \emph{safe subset function} if it is a weak safe subset function and \(H(x,x) \ggeq 0\) for all \(x \in S\).
\end{define}

\begin{define}
    \label{def:SDTE-MCBF}
    A safe subset function \(H : \R^n \times \R^n \to \mathbb{S}^p\) is called a 
    Subset-based Discrete-Time Exponential Matrix Control Barrier Function (SDTE-MCBF)
    for the system \eqref{eq:discrete_dynamics} if there exists a constant \(0 < \gamma \leq 1\) and a domain \(\mathcal{D} \supset S\) such that 
    for all \(x \in \mathcal{D}\) there exists a control input \(u \in \R^m \) satisfying
    \begin{align}
        H(f(x) + g(x) u, x) \ggeq (1-\gamma)H(x,x).
    \end{align}
\end{define}
When the safe subset described by \(H\) is obtained through projection-based methods (as described previously in this section), we will refer to \(H\) as a \emph{Projection-based Discrete-Time Exponential Matrix Control Barrier Function} (PDTE-MCBF). A PDTE-MCBF is a special case of SDTE-CBFs since it is possible that a safe subset could be obtained by means other than projections. However, we leave further investigation of this distinction for future work.


Given a safe subset function \(H\), consider the following set of control inputs:
{\small
\begin{align}
    &K_{H}(x_k) \triangleq \\
    &\hspace{2em}\{u \in \R^m : H(x_{k+1}, x_k) \ggeq (1-\gamma) H(x_k, x_k)\} \nonumber \\
    &= \{u \in \R^m : H(f(x_k) + g(x_k) u, x_k) \ggeq (1-\gamma) H^\epsilon(x_k, x_k) \}
\end{align}
}
The following theorem demonstrates that \(K_{H}(x_k)\) represents the set of control inputs \(u\) that render the set \(S\) forward invariant.

\begin{theorem}
    \label{thm:forward_invariant}
    Let \(H\) be an SDTE-MCBF on a domain \(\mathcal{D} \supset S\) for the system \eqref{eq:discrete_dynamics}.
    Then any control law \(u(x_k)\) such that \(u(x_k) \in K_{H}(x_k)\) for all \(k \geq 0\) renders the safe set \(S\) forward invariant.
\end{theorem}

\begin{proof}
    Suppose \(x_0 \in S\).
    The fact that \(H\) is a SDTE-MCBF implies that it is a safe subset function by Definition \ref{def:SDTE-MCBF}, which implies \(H(x_0, x_0) \ggeq 0\) by Definition \ref{def:safe_subset_function}. The fact that \(u(x_0) \in K_{H}(x_0)\) implies that \(H(f(x_0) + g(x_0)u(x_0) \ggeq (1-\gamma) H(x_0,x_0) \ggeq 0 \). It follows that \(x_{1} \in S\).
    Proceeding inductively, assume that \(x_k \in S\). By similar arguments it follows that \(x_{k+1} \in S\) and \(H(f(x_k) + g(x_k)u(x_k),x_k) \ggeq (1-\gamma)H(x_k,x_k) \ggeq 0\), which concludes the proof.
\end{proof}

Given these definitions, we now formally prove that the function \(H^\epsilon\) can be classified as a safe subset function.

\begin{lemma}
    Let \(S\) be nonempty, and suppose \(\overline{S}\) satisfies Assumption \ref{assume:unsafe_set}. Let \(H^\epsilon\) be defined as in \eqref{eq:H_prime_set}.
    Choose \(\mathcal{D} = S\) and define \(S^\epsilon = \{z \in S : \textup{dist}(x,\overline{S}) \geq \epsilon \} \subset \mathcal{D}\), where the distance is defined in terms of the Euclidean norm.
    Then \(H^\epsilon\) is a safe subset function for \(S^\epsilon\).
\end{lemma}

\begin{proof}
    By the properties of convex sets and Lemma \ref{lem:unsafe_inequality}, \(a_i(x)^\intercal z - b_i^\epsilon(x) \geq 0 \implies \pth{\frac{x - P_{\overline{S}_i}(x)}{\nrm{x - P_{\overline{S}_i}(x)}}}^\intercal \pth{z - P_{\overline{S}_i}(x)} \geq \epsilon\) implies that \(\text{dist}(z,\overline{S}_i) \geq \epsilon\) for all \(i=1,\ldots,Q\). By equation \eqref{eq:H_prime_set}, it follows that all points in the set \(\{z \in \R^n : H^\epsilon(z,x) \ggeq 0\}\) satisfy \(\text{dist}(z,\overline{S}) \geq \epsilon\), and therefore \(\{z \in \R^n : H^\epsilon(z,x) \ggeq 0\} \subset S^\epsilon\). The function \(H^\epsilon\) is therefore a weak safe set function.
    
    Next, choose any \(x \in S^\epsilon\). It holds that \(\nrm{x - P_{\overline{S}_i}(x)}_2 \geq \epsilon\), which implies that \(a_i(x)^\intercal x - b_i^\epsilon(x) \geq 0\) for all \(i=1,\ldots,Q\). It follows that \(H^\epsilon(x,x) \ggeq 0\) for all \(x \in S^\epsilon\).
\end{proof}

Since the definition of SDTE-MCBF relies upon the specific dynamics under consideration, proofs that \(H^\epsilon\) form a SDTE-MCBF will require case-by-case analysis.
Note that any \(H^\epsilon\) as defined in \eqref{eq:H_prime_set} that qualifies as a SDTE-MCBF is also a PDTE-MCBF since the subset of the safe set is obtained via projections.
Under the assumption that \(H^\epsilon\) is an SDTE-MCBF (or PDTE-MCBF) and given a nominal control input \(u_\text{nom}(x_k)\) and the projected points \(\{ P_{\overline{S}_i}\pth{x_k} \}_{i=1}^Q\), it is possible to compute a safety-preserving control input \(u^*(x_k) \in K_H^\epsilon(x_k)\) using the following convex QP:
\begin{subequations}
\label{eq:safety_preserving_u}
\begin{alignat}{2}
    &u^*(x_k) = \underset{u}{\arg\min} \nrm{u - u_\text{nom}} \\
     & \text{s.t. }  H^\epsilon(f(x_k) + g(x_k)u, x_k) \ggeq (1-\gamma) H^\epsilon(x_k,x_k)
\end{alignat}
\end{subequations}


\begin{algorithm}
\caption{Safe Input Computation for PDTE-MCBF}
\label{alg:main_algorithm}
    \textbf{Inputs:} State \(x_k\), nominal control input \(u_\text{nom}(x_k)\), parameter \(\epsilon > 0\)\;
    \textbf{Output:} Safety-preserving control input \(u^*(x_k)\)\;
        \For{\(i=1,\ldots,Q\)}{
            \(P_{\overline{S}_i} \gets  \arg\min_{z \in \R^n} \nrm{z - x_k} \text{s.t. } \overline{H}_i(x) \gleq 0 \)\;
            \hspace{1em}(See Eqs. \eqref{eq:convex_projection} or \eqref{eq:proj_mat_constraint})
        }
        \(u^*(x_k) \gets \arg\min_{u \in \R^m} \nrm{u - u_\text{nom}}  \text{ s.t. } H^\epsilon(f(x_k)+g(x_k)u, x_k) \ggeq 0 \)\;
        \hspace{1em}(See Eqs. \eqref{eq:h_prime}, \eqref{eq:H_prime_set}, \eqref{eq:safety_preserving_u})\;
    \textbf{return} \(u^*(x_k)\)
\end{algorithm}

Algorithm \ref{alg:main_algorithm} summarizes the process of computing safe control inputs \(u^*(x_k)\) for the form of PDTE-MCBFs defined in this paper.
For simplicity, algorithm \ref{alg:main_algorithm} shows each projection point \(P_{\overline{S}_i}(x_k)\) being computed by a separate optimization problem instance. However, we note that all projection points \(\vec{P} \triangleq \bmx{P_{\overline{S}_1} & \cdots & P_{\overline{S}_Q}}^\intercal \) can be computed simultaneously in a single optimization instance as follows:
\begin{subequations}
\label{eq:proj_mat_constraint}
\begin{align}
    \vec{P} = \underset{\vec{z} \triangleq \bmxs{z_1,\cdots,z_Q}^\intercal}{\arg\min} &&& \sum_{i=1}^Q \nrm{z_i - x_k} \\
    \text{s.t.} &&& \overline{H}_i(z_i) \gleq 0\ \forall i=1,\ldots,Q. 
\end{align}
\end{subequations}

We point out that our proposed method only requires solving two sequential convex optimization problems per time step \(k\), the second of which is a QP. 
In special cases where \(\overline{S}_i\) takes a specific form such as a circle or ellipse, it is possible to obtain the projection point directly in closed-form.
Compared to prior nonconvex formulations for discrete-time CBFs, our convex optimization method has polynomial-time computational complexity and guarantees convergence to globally optimal solutions.


\subsection{Zeroing Properties of the Projection Method}

The original notion of discrete-time exponential control barrier functions renders safe sets exponentially stable in the state space \cite{agrawal2017discrete}. In other words, if the initial state \(x_0\) is in the unsafe set \(\overline{S}\), a safe control input derived from a discrete-time exponential CBF causes \(x_k\) to converge exponentially to the safe set \(S\). 
The convex projection method under Assumption \ref{assume:unsafe_set} in general operates only when the state \(x_k\) does not coincide with any of the projections \(P_{\overline{S}_i}(x_k)\) onto the unsafe sets \(\overline{S}_i\), \(i=1,\ldots,Q\). If \(x_k \in \overline{S}_i\), then \(P_{\overline{S}_i}(x_k) = x_k\) and the prior formulas for \(a_i, b_i\) result in ill-defined constraints.
Deriving methods to ensure \(x_k \in \overline{S}\) converges exponentially to \(S\) is more difficult in general for the projection method proposed previously; however  
this section presents methods to render \(S\) asymptotically stable when \(x_k \in \overline{S}\).

If \(x_k \in \partial \overline{S}_i\), we may use the normal cone \(N_{\overline{S}_i}(x_k)\) to compute a safe control input. Choose any \(x_N(x_k) \in N_{\overline{S}_i}(x_k)\) such that \(x_N(x_k) \neq x_k\) and \( x_N(x_k) - x_k \in N_{\overline{S}_i}(x_k)\) . Define the 
functions \(\widehat{a}_i^{\partial} : \R^n \to \R^n, \widehat{b}_i^{\partial,\epsilon} : \R^n \to \R\) as
\begin{subequations}
\label{eq:normal_cone_ai_bi}
\begin{align}
    \widehat{a}_i^\partial(x_k) &= \frac{x_N(x_k) - x_k}{\nrm{x_N(x_k) - x_k}}, \\
    \widehat{b}_i^{\partial, \epsilon}(x_k) &= \epsilon + \frac{(x_N(x_k) - x_k)}{\nrm{x_N(x_k) - x_k}}^\intercal x_k.
\end{align}
\end{subequations}
We can then define the safe set function
\begin{align}
    \widehat{h}_i^{\partial, \epsilon}(z, x_k) \triangleq \pth{\widehat{a}_i^\partial(x_k)}^\intercal z - \widehat{b}_i^{\partial, \epsilon}(x_k).
\end{align}

\begin{lemma}
    \label{lem:unsafe_inequality_normal_cone}
    Let \(x_k \in \partial \overline{S}_i\) for all \(i=1,\ldots,Q\). Let \(\epsilon > 0\).
    Then the set \(\{z \in \R^n : h_i^{\partial, \epsilon}(z, x_k) \geq 0 \}\) is disjoint from \(\overline{S}_i \) for all \(i=1,\ldots,Q\).
\end{lemma}
\begin{proof}
    By definition, \( x_N(x_k) - x_k \in N_{\overline{S}_i}(x_k)\). The result follows using similar arguments as Lemma \ref{lem:unsafe_inequality} and Lemma \ref{lem:convex_subset}
\end{proof}
The safe set function \(\widehat{h}_i^{\partial, \epsilon}\) may be used in place of \(h_i^\epsilon\) in Algorithm \ref{alg:main_algorithm} when \(x_k \in \partial \overline{S}_i\).

The more difficult case is when the state is in the interior of the unsafe set; i.e., \(x_k \in \text{int}\pth{\overline{S}_i}\). In this case the normal cone is trivially the zero vector: \(N_{\overline{S}_i}(z) = \{0\}\ \forall z \in \overline{S}_i\).
We therefore use the notion of \emph{depth} of a convex set. The depth function \(\text{depth} : \R^n \times 2^{\R^n} \to \R\) is defined as
\begin{align}
    \text{depth}\pth{x,C} \triangleq \text{dist}(x,\R^n \backslash S),
\end{align}
where \(C\) is a subset of \(\R^n\). For convex sets such as \(\overline{S}_i\) under Assumption \ref{assume:unsafe_set}, the depth can be computed as \cite[Sec. 8.5.1]{boyd2004convex}
\begin{align}
    \text{depth}(x_k,\overline{S}_i) = \max_R &&& R \label{eq:general_depth} \\
    \text{s.t.} &&& \mathfrak{g}_j(x,R) \leq 0\ j=1,\ldots,J', \nonumber
\end{align}
This optimization problem is convex, but the functions \(\mathfrak{g}_j\) are difficult to evaluate in general. However, the optimization problem in \eqref{eq:general_depth} is tractable for specific classes of convex functions. When \(\overline{S}_i\) is a polytope in halfspace form \(a_{i,j}'^\intercal x_{k} \leq b_{i,j}'\), \(\text{depth}(x_k,\overline{S}_i)\) can be computed via the following linear program \cite[Sec. 8.5.1]{boyd2004convex}:
\begin{align}
    \text{depth}(x_k,\overline{S}_i) = \max_R &&& R \\
    \text{s.t.} &&& a_{i,j}'^\intercal x_k + R\nrm{a_{i,j}'} \leq b_{i,j}
\end{align}
When 
\(\overline{S}_i\) takes the form of an intersection of ellipsoids defined by quadratic inequalities \(\overline{S}_i = \{x : x^\intercal M_{i,j} x + 2 v_{i,j}^\intercal x + d_{i,j} \leq 0,\ j=1,\ldots,J' \), \(M_{i,j} \in \R^{n\times n}\), \(v_{i,j} \in \R^n\), \(d_{i,j} \in \R\), the depth can be computed via the following convex SDP \cite[Sec. 8.5.1]{boyd2004convex}:
\begin{align*}
    &\text{depth}(x_k, \overline{S}_i) = \\
    &\hspace{1em} \max_{R\in \R,\ \lambda_{i,j} \in \R} R\\
    &\hspace{1em} \text{s.t.} \bmxs{
        -\lambda_{i,j} - d_{i,j} + v_{i,j}^\intercal M_{i,j}^{-1} v_{i,j} & 0 & (x + M_{i,j}^{-1} v_{i,j})^\intercal \\
        0 & \lambda_{i,j} I & R I \\
        x_k + M_{i,j}^{-1} v_{i,j} & RI & M_i^{-1}
    } \ggeq 0, \\
    &\hspace{1em} j=1,\ldots,J'
\end{align*}
Due to the nature of the objective function \eqref{eq:general_depth}, the value of \(\text{depth}(x_k, \overline{S}_i)\) is equal to the value of the optimal point \(R^*(x_k)\).
Using the implicit function theorem, it is possible to obtain a descent direction for the depth by computing \(\frac{\partial}{\partial x} \text{depth}(x_k, \overline{S}_i) = \frac{\partial R^*(x_k)}{\partial x}\) using standard methods to backpropagate through convex optimization programs \cite{agrawal2019differentiable}. The following halfspaces can then be defined:
\begin{subequations}
\label{eq:ai_bi_R}
\begin{align}
    a_i^R(x_k) &\triangleq \frac{-\partial R_i^*(x_k) / \partial x}{\nrm{\partial R_i^*(x_k) / \partial x}}, \\
    b_i^{R}(x_k) &\triangleq \pth{\frac{-\partial R_i^*(x_k) / \partial x}{\nrm{\partial R_i^*(x_k) / \partial x}}}^\intercal x_k - R_i^*(x_k)
\end{align}
\end{subequations}

We can then define the function 
\begin{align}
h_i^{R}(z,x_k) = a_i^R(x_k)^\intercal z - b_i^{R}(x_k), \label{eq:h_i_R}
\end{align}
which may be used in place of \(h_i^\epsilon\) in Algorithm \ref{alg:main_algorithm} when \(x_k \in \textup{int}(\overline{S}_i)\).

\begin{lemma}
    \label{lem:depth_decrease}
    Let \(\overline{S}_i\) be defined as in Assumption \ref{assume:unsafe_set}.
    Let \(h_i^R\) be defined as in \eqref{eq:h_i_R}.
    Let \(x_k \in \textup{int}(\overline{S}_i)\) and
    let \(z \in \R^n\) be any vector \(z\) satisfying \(h_i^R(z,x_k) \geq (1-\gamma) h_i^R(x_k,x_k)\) for some \(0 < \gamma \leq 1\).
    Then \(\textup{depth}(z,\overline{S}_i) < \textup{depth}(x_k, \overline{S}_i)\).
\end{lemma}

\begin{proof}
    By definition, \(R_i^*(x_k) = \textup{depth}(x_k, \overline{S}_i)\).
    The vector \(\frac{-\partial R_i^*(x_k) / \partial x}{\nrm{\partial R_i^*(x_k) / \partial x}}\) can be considered an element of the normal cone to the convex set \(C_i(x_k) = \{y \in \R^n : R_i^*(x_k) \leq R_i^*(y) \}\), where \(x_k \in \partial C_i(x_k)\).
    Observe that \(h_i^R(x_k,x_k) = R_i^*(x_k)\).
    The equation \(h_i^R(z,x_k) \geq h_i^R(x_k,x_k\) can therefore be rewritten as
    \begin{align}
        \frac{-\partial R_i^*(x_k) / \partial x}{\nrm{\partial R_i^*(x_k) / \partial x}}^\intercal \pth{z - x_k} \geq (1-\gamma) R_i^*(x_k).
    \end{align}
    Since \(R_i^*(x_k) > 0\) for \(x_k \in \overline{S}_i\) and \(0 < \gamma \leq 1\), we therefore have \(R_i^*(z) < R_i^*(x_k)\), or equivalently \(\textup{depth}(z, \overline{S}_i) < \textup{depth}(x_k, \overline{S}_i)\). 
\end{proof}

Finally, observe that \(h_i^R(f(x) + g(x) u,x_k)\) is affine and therefore concave in \(u\), which empowers rapid computation of feasible control inputs \(u\) using convex optimization techniques.

\subsection{Indefinite Matrix Safe Sets}

The recent work \cite{ong2025matrix} introduced the notion of \emph{indefinite safe sets} defined as
\begin{align}
    S &= \{x \in \R^n : H(x) \not\prec 0\},\\
    &= \{x \in \R^n : \lambda_p(H(x)) \geq 0 \}.
\end{align}
Here, recall that by convention the eigenvalues of \(H\) are ordered \(\lambda_1(H(x)) \leq \lambda_2(H(x)) \leq \cdots \leq \lambda_p(H(x))\).
Indefinite safe sets have application to enforcing disjunctive boolean constraints on multiple CBFs \cite{ong2025matrix}. We generalize this notion of indefinite safe sets as follows:
\begin{align}
    S^{(j)} &= \{x \in \R^n : \lambda_j(H(x)) \geq 0 \},\ 1 \leq j \leq p. \label{eq:multi_indef_safe_set}
\end{align}
In words, \(S^{(j)}\) represents the states where eigenvalues \(j\) through \(p\) are non-negative. If we let \(r = p - j\), then this represents enforcing \(r\) out of \(p\) constraints to be active at any time. We next show that the set \(S^{(j)}\) may be rendered forward invariant by enforcing the following constraint at every time step \(k \geq 0\) for some \(0 < \gamma \leq 1\):
\begin{align}
    H(x_{k+1}) - H(x_k) \ggeq  -\gamma \lambda_{j}(H(x_k)) I_{p \times p}.
\end{align}
Given the dynamics \eqref{eq:discrete_dynamics}, define the set-valued mapping \(K_H^{(j)} : \R^n \to 2^{\R^m} \) as
{\small
\begin{align}
    &K_H^{(j)}(x_k) \triangleq \big\{u \in \R^m : \label{eq:K_H_choose} \\
    &H(f(x) + g(x)u) \ggeq H(x_k) - \gamma \lambda_j(H(x_k))I_{p\times p}) \big\}. \nonumber
\end{align}
}

\begin{theorem}
    Let \(S^{(j)}\) be defined as in \eqref{eq:multi_indef_safe_set} and let \(K_H^{(j)}\) be defined as in \eqref{eq:K_H_choose}. Let \(x_0 \in S^{(j)}\). Then any control input \(u(x_k) \in K_H^{(j)}(x_k)\) \(\forall k \geq 0\) renders the set \(S^{(j)}\) forward invariant.
    Furthermore, such a \(u(x_k)\) renders the set \(S^{(j)}\) exponentially stable in \(\R^n\).
\end{theorem}

\begin{proof}
    For brevity denote \(H_{k} = H(x_{k})\).
    Define
    \(\Delta^{(j)}H_k \triangleq H_{k+1} - \pth{H_k - \gamma \lambda_j(H_k) I_{p\times p}}\). Using Weyl's inequality (Theorem \ref{thm:weyl}) and setting \(i=j\), we have
    \begin{align*}
        &\lambda_{1}\pth{\Delta^{(j)}H_k} + \lambda_j\Big(H_k - \gamma \lambda_j(H_k) I_{p\times p}\Big) \leq \\
        &\hspace{2em}\lambda_j\pth{\Delta^{(j)}H_k + H_k - \gamma \lambda_j(H_k) I_{p\times p}}, \\
        &\hspace{2em}= \lambda_j\pth{H_{k+1}}.
    \end{align*}
    However, \(u(x_k) \in K_H^{(j)}(x_k)\) implies that \(\lambda_{1}(\Delta^{(j)}H_k) \geq 0\). It follows that
    \begin{align*}
        \lambda_j\Big(H_k - \gamma \lambda_j(H_k) I_{p\times p}\Big) &\leq \lambda_j(H_{k+1}), \\
        \implies (1-\gamma) \lambda_j(H_k) &\leq \lambda_j(H_{k+1}).
    \end{align*}
    By induction, we have
    \begin{align}
        \lambda_j(H_{k+1}) \geq (1-\lambda)^k \lambda_j(H_0).
    \end{align}
    It follows that \(S^{(j)}\) is both forward invariant and exponentially stable in \(\R^n\).
\end{proof}
The set of control inputs within \(K_H^{(j)}(x_k)\) ensure that the largest eigenvalues \(\lambda_j\) through \(\lambda_p\) remain positive, but the smallest eigenvalues \(\lambda_1\) through \(\lambda_{j-1}\) are free to remain negative. This behavior relates to recent work on combinatorial ``\(p\)-choose-\(r\)" CBFs \cite{ong2025combinatorial}, however our contributions here apply to discrete-time CBFs. Given a nominal control input \(u_\textup{nom}(x_k)\), control inputs within \(K_H^{(j)}(x_k)\) may be computed using the following optimization formulation:
\small
\begin{align}
    &u^{(j)*}(x_k) = \min_{u \in \R^m} \hspace{1em} \nrm{u - u_\textup{nom}(x_k)} \\ 
    &\textup{s.t. } H(f(x_k) + g(x_k)u) - H(x_k) + \gamma \lambda_j(H(x_k))I_{p\times p} \ggeq 0. \nonumber
\end{align}
\normalsize
The results in preceding sections may be applied to ensure that this optimization formulation is convex. 

\section{Numerical Simulations}
\label{sec:simulations}

In this section, the PDTE-MCBF technique developed in the previous section and shown in Algorithm \ref{alg:main_algorithm} is implemented in numerical examples.
This section also provides a performance comparison between the proposed PDTE-MCBF technique and nonconvex CBF formulations.
Note that in this work, the numerical examples are constrained to a plane in which only horizontal and vertical motion is permitted and thus only two-dimensional obstacles are considered.
Consequently, the subindices $\rmh$ and $\rmv$ used in the variables to denote a connection to horizontal and vertical coordinates, respectively.
For instance, $p_{\rmh}$ and $p_{\rmv}$ are used in these examples to denotes positions in the horizontal and vertical directions, respectively.

First, an overview of the obstacles considered in the numerical examples is provided, in which the unsafe sets, matrix convex functions, the projected point calculations, and the safe sets associated with each obstacle are discussed; in particular, analogous parameters to the previously defined buffer distance $\epsilon$ are defined for the safe sets of each obstacle.
Then, two systems are considered: a system in a plane composed of two double integrators and the outer-loop controller of a bicopter lateral flight system.
The nominal controllers in the examples are given by reference tracking LQR controllers. The details of this implementation are omitted for brevity and since the main focus of the paper is the performance of the CBF.
The implementation of the PDTE-MCBF technique requires the solution of a constrained linear least-squares optimization problem, as shown in \eqref{eq:safety_preserving_u}. 
In all examples, this problem is solved by using the \texttt{lsqlin} solver from Matlab with the \texttt{active-set} algorithm.
The solver for nonconvex CBF formulations is discussed in each numerical example.

Furthermore, in this work, the safe sets for obstacle avoidance depend only on the position states.
Hence, there exists $C_{\rm pos} \in \BBR^{2 \times n},$ such that $C_{\rm pos} x \in \BBR^2$ yields the position state vector related to state $x.$
Thus, \eqref{eq:a_i_b_i} is rewritten as
\footnotesize
\begin{subequations}
\label{eq:a_i_b_i_pos}
\begin{align}
    a_i(x_k) &\triangleq \frac{C_{\rm pos}^\rmT C_{\rm pos} (x_k - P_{\overline{S}_i}(x_k))}{\nrm{C_{\rm pos} (x_k - P_{\overline{S}_i}(x_k))}_2}, \\
    b_i^\epsilon (x_k) &\triangleq \epsilon +  \pth{\frac{C_{\rm pos} (x_k - P_{\overline{S}_i}(x_k))}{\nrm{C_{\rm pos}(x_k - P_{\overline{S}_i}(x_k))}_2}}^\intercal (C_{\rm pos} P_{\overline{S}_i}\pth{x_k}).
\end{align}
\end{subequations}
\normalsize
The matrix convex functions presented in the obstacle subsection are also defined in terms of the $C_{\rm pos}$ operator.

All computational results in this paper were obtained using a laptop PC running
Windows 10 Pro, version 22H2, 
OS build 19045.6332,
with an Intel Core i7-10750H processor running at 2.60 GHz 
and a 32GB, 3200 MHz RAM,
with MATLAB version R2024b Update 5.

\subsection{Obstacles} \label{subsec:ex_obstacles}

The four obstacles considered in this work are a circle, an ellipse, a convex polytope, and a spectrahedron, which are shown in Figure \ref{fig:ex_obstacles}.
Their properties and the projected point calculation procedure and the CBF associated with each of these obstacles are presented next.

\subsubsection{Circle}
\label{ssec:circle}
A circle $\mathfrak{C}$ is defined by a center position $p_{\mathfrak{C}} \in \BBR^2$ and a radius $r_{\mathfrak{C}}.$
For a circle $\mathfrak{C}$ obstacle, the unsafe set is given by
\begin{equation}
    \overline{S}_{\mathfrak{C}} = \{ p \in \BBR^2 \colon \lVert p - p_{\mathfrak{C}} \rVert_2 \le r_{\mathfrak{C}} \}, \label{eq:unsafe_set_circle}
\end{equation}
and thus, the corresponding matrix convex function describing the unsafe set is given by
\begin{equation}
    \overline{H}_{\mathfrak{C}} (x) = \lVert C_{\rm pos} x - p_{\mathfrak{C}} \rVert_2 - r_{\mathfrak{C}}. \label{eq:unsafe_set_matrix_circle}
\end{equation}
Note that the projection of a state vector $x_k$ onto the set $\overline{S}_{\mathfrak{C}}$, given by \eqref{eq:convex_projection} can be expressed in closed-form as
\small
\begin{equation}
    P_{\overline{S}_{\mathfrak{C}}} (x_k) =  \frac{\min\{r_{\mathfrak{C}},  \lVert C_{\rm pos} x_k - p_{\mathfrak{C}} \rVert_2\}}{\lVert C_{\rm pos} x_k - p_{\mathfrak{C}} \rVert_2} (C_{\rm pos} x_k - p_{\mathfrak{C}}) + p_{\mathfrak{C}}, \label{eq:proj_circle}
\end{equation}
\normalsize
which can be used to calculate the projected point without solving an optimization problem.
Finally, it follows from the unsafe set \eqref{eq:unsafe_set_circle} that the CBF associated with the circle obstacle safe set $S_{\mathfrak{C}}$ with buffer distance $\epsilon_{\mathfrak{C}}$ is given by
\begin{equation}
    h_{\mathfrak{C}} (x_k) = (r_{\mathfrak{C}} + \epsilon_{\mathfrak{C}})^2 - (C_{\rm pos } x_k - p_{\mathfrak{C}})^\rmT (C_{\rm pos } x_k - p_{\mathfrak{C}}), \label{eq:h_circle}
\end{equation}
which can be used to implement the safe control input constraints shown in \eqref{eq:nonconvex_constraint}.
Note that $\epsilon_{\mathfrak{C}}$ increases the radius of the unsafe set considered by the CBF.

\subsubsection{Ellipse}
\label{ssec:ellipse}
An ellipse $\mathfrak{El}$ is defined by a center position $p_{\mathfrak{El}} \in \BBR^2,$ the semi-major axis unit vector $v_{\mathfrak{El}} \in \BBR^2,$ such that $\lVert v_{\mathfrak{El}} \rVert_2 = 1,$ and the semi-major and semi-minor axes lengths  $\ell_{\mathfrak{El}, 1} \ge \ell_{\mathfrak{El}, 2}  > 0 $, respectively.
Let $Q_{\mathfrak{El}} \isdef \matl v_{\mathfrak{El}} &  v_{\mathfrak{El}, \perp} \matr \in \BBR^{2 \times 2},$ where $v_{\mathfrak{El}, \perp} \isdef \matl 0 & -1 \\ 1 & 0 \matr v_{\mathfrak{El}},$ and let
\begin{equation*}
    M_{\mathfrak{El}} \isdef Q_{\mathfrak{El}} \matl \ell_{\mathfrak{El}, 1}^{-2} & 0 \\ 0 & \ell_{\mathfrak{El}, 2}^{-2} \matr Q_{\mathfrak{El}}^\rmT.
\end{equation*}
Then, for an ellipse $\mathfrak{El}$ obstacle, the unsafe set is given by
\begin{equation}
    \overline{S}_{\mathfrak{El}} = \{ p \in \BBR^2 \colon (p - p_{\mathfrak{El}})^\rmT M_{\mathfrak{El}} (p - p_{\mathfrak{El}})  \le 1 \}, \label{eq:unsafe_set_ellipse}
\end{equation}
and thus, the corresponding matrix convex function is given by
\begin{equation}
    \overline{H}_{\mathfrak{El}} (x) =  (C_{\rm pos} x - p_{\mathfrak{El}})^\rmT M_{\mathfrak{El}} (C_{\rm pos} x - p_{\mathfrak{El}}) - 1. \label{eq:unsafe_set_matrix_ellipse}
\end{equation}
Note that the projection of a state vector $x_k$ onto the set $\overline{S}_{\mathfrak{El}}$, given by \eqref{eq:convex_projection} can be formulated as a second-order cone optimization problem;
in this work, this optimization problem is solved by using the \texttt{coneprog} solver from Matlab.
Finally, it follows from the unsafe set \eqref{eq:unsafe_set_ellipse} that the CBF associated with the ellipse obstacle safe set $S_{\mathfrak{El}}$ with buffer distance $\epsilon_{\mathfrak{El}}$ is given by
\begin{equation}
    h_{\mathfrak{El}} (x_k) = 1 - (C_{\rm pos} x_k - p_{\mathfrak{El}})^\rmT M_{\mathfrak{El}, \epsilon_{\mathfrak{El}}} (C_{\rm pos} x_k - p_{\mathfrak{El}}), \label{eq:h_ellipse}
\end{equation}
where
\small
\begin{equation*}
    M_{\mathfrak{El}, \epsilon_{\mathfrak{El}}} \isdef Q_{\mathfrak{El}} \matl (\ell_{\mathfrak{El}, 1} + \epsilon_{\mathfrak{El}})^{-2} & 0 \\ 0 & (\ell_{\mathfrak{El}, 2} + \epsilon_{\mathfrak{El}})^{-2} \matr Q_{\mathfrak{El}}^\rmT,
\end{equation*}
\normalsize
and which can be used to implement safe control input constraints shown in \eqref{eq:nonconvex_constraint}.
Note that $\epsilon_{\mathfrak{El}}$ increases the lengths of the semi-major and semi-minor aces of the ellipse corresponding to the unsafe set considered by the CBF.

\subsubsection{Convex Polytope}
\label{ssec:polytope}
A convex polytope $\mathfrak{P}$ is defined by a matrix composed by $n_\mathfrak{P}$ vertices, given by $v_{\mathfrak{P}} \in \BBR^{2 \times n_\mathfrak{P}}.$
Then, for a polytope $\mathfrak{P}$ obstacle, the unsafe set is given by
\begin{equation}
    \overline{S}_{\mathfrak{P}} = \left\{ \bigcap_{i = 1}^{n_\mathfrak{P}} \{p \in \BBR^2 \colon -A_{\mathfrak{P}, i} p - b_{\mathfrak{P}, i} \le 0\} \right\}, \label{eq:unsafe_set_polytope}
\end{equation}
where $-A_{\mathfrak{P}, i} \in \BBR^{1 \times 2}, -b_{\mathfrak{P}, i} \in \BBR$ are associated with the halfspace that is determined by the $i$-th edge of $\mathfrak{P}$ and that does not contain the mean of all $n_\mathfrak{P}$ vertices.
Hence, the corresponding matrix convex function is given by
\begin{align}
    \overline{H}_{\mathfrak{P}} (x) = 
    {\rm diag} &(-A_{\mathfrak{P}, 1} C_{\rm pos} x - b_{\mathfrak{P}, 1}, \ldots, \nn \\ 
    &-A_{\mathfrak{P}, n_\mathfrak{P}}  C_{\rm pos} x - b_{\mathfrak{P}, n_\mathfrak{P}}). \label{eq:unsafe_set_matrix_polytope}
\end{align}
Note that the projection of a state vector $x_k$ onto the set $\overline{S}_{\mathfrak{P}}$, given by \eqref{eq:convex_projection} can be written as
\begin{subequations}
\label{eq:proj_polytope}
\begin{alignat}{2}
    P_{\overline{S}_{\mathfrak{P}}} (x_k) = &&& \underset{z \in \BBR^n}{\arg\min}  \nrm{z - C_{\rm pos} x_k} \\
     \text{s.t. } &&\hspace{1em} & \begin{matrix} v_{\mathfrak{P}} \lambda_{\mathfrak{P}} = z, \\
    \lambda_{\mathfrak{P}} \in \BBR^{n_\mathfrak{P}}, \lambda_{\mathfrak{P}} \ge 0, \mathds{1}_{1 \times n_\mathfrak{P}} \lambda_{\mathfrak{P}} = 1,\end{matrix}
\end{alignat}
\end{subequations}
and thus can be formulated as a constrained linear least-squares problem; 
in this work, this optimization problem is solved by using the \texttt{lsqlin} solver from Matlab with the \texttt{active-set} algorithm option.
This problem is warm-started by providing the closest vertex to $C_{\rm pos} x_k$ as an initial solution.
Next, it follows from the unsafe set \eqref{eq:unsafe_set_polytope} that the safe set corresponding to the polytope obstacle is given by
\begin{equation}
    S_{\mathfrak{P}} = \left\{ \bigcup_{i = 1}^{n_\mathfrak{P}} \{p \in \BBR^2 \colon A_{\mathfrak{P}, i} p + b_{\mathfrak{P}, i} \ge 0\} \right\}.
\end{equation}
In order to derive a CBF for this union of sets, a discrete-time version of the CBF constraint proposed in (15) from \cite{ong2025combinatorial} is formulated.
Using a combinatorial CBF was required for the implementation of  the CBF associated with the polytope, and we will explore this connection in more detail in future work.
Hence, the safe input constraints associated with the polytope safe set $S_{\mathfrak{P}}$ with buffer distance $\epsilon_{\mathfrak{P}}$ that replace \eqref{eq:nonconvex_constraint} are given by
\begin{align}
    h_{\mathfrak{P}, i} & (f(x_k) + g(x_k)u) \nn \\ &\ge h_{\mathfrak{P}, i} (x_k) - \gamma \left( h_{\mathfrak{P}} (x_k) + \lvert h_{\mathfrak{P}, i} (x_k) - h_{\mathfrak{P}} (x_k) \rvert \right) \nn \\ &\forall i \in \{1, \ldots, n_\mathfrak{P}\}, \label{eq:nonconvex_constraint_polytope}
\end{align}
where $\gamma \in [0, 1],$
\begin{align}
     h_{\mathfrak{P}, i} (x_k) &\isdef A_{\mathfrak{P}, i} C_{\rm pos} x_k +  b_{\mathfrak{P}, i} + \eta_i \epsilon_\mathfrak{P} \\
     & \hspace{5em} \forall i \in \{1, \ldots, n_\mathfrak{P}\}, \nn \\
    h_{\mathfrak{P}, {\rm max}} (x_k) &\isdef \max_{i \in \{1, \ldots , n_\mathfrak{P}\}} h_{\mathfrak{P}, i} (x_k),
\end{align}
where, for all $i \in \{1, \ldots, n_\mathfrak{P}\}$, $\eta_i \in \{-1, 1\}$ is chosen so that the distance between the hyperplane given by $h_{\mathfrak{P}, i}$ and the center of $\mathfrak{P}$ is increased by $\epsilon_\mathfrak{P}$ relative to the case where $\eta_i = 0.$

\subsubsection{Spectrahedron}
\label{ssec:spectra}
A spectrahedron $\mathfrak{Sp}$ in $\BBR^2$ is defined by a center position $p_{\mathfrak{Sp}} \in \BBR^2,$ matrix size $n_{\mathfrak{Sp}}$, symmetric matrices $A_{0, \mathfrak{Sp}}, A_{1, \mathfrak{Sp}}, A_{2, \mathfrak{Sp}} \in \BBR^{ n_{\mathfrak{Sp}} \times n_{\mathfrak{Sp}}},$ and rotation angle $\theta_{\mathfrak{Sp}}$ in radians.
Then, for a polyhedron $\mathfrak{Sp}$ obstacle, the unsafe set is given by
\small
\begin{align}
    \overline{S}_{\mathfrak{Sp}} &= \left\{p \in \BBR^2 \colon A_{0, \mathfrak{Sp}} + z_1 A_{1, \mathfrak{Sp}} + z_2 A_{2, \mathfrak{Sp}} \succeq 0, \right. \nn \\
    & \quad \quad \left. \mbox{where } \matl z_1 \\ z_2  \matr = R_{\theta_{\mathfrak{Sp}}} (p - p_{\mathfrak{Sp}}) \right\}, \label{eq:unsafe_set_spectrahedron}
\end{align}
\normalsize
where $R_{\theta_{\mathfrak{Sp}}} \in \BBR^{2 \times 2}$ is the rotation matrix associated with $\theta_{\mathfrak{Sp}}.$
Hence, the corresponding matrix convex function is given by
\small
\begin{align}
    \overline{H}_{\mathfrak{Sp}} (x) &= - A_{\mathfrak{Sp}} \left(\matl 1 \\ R_{\theta_{\mathfrak{Sp}}} (C_{\rm pos} x - p_{\mathfrak{Sp}}) \matr \otimes I_{ n_{\mathfrak{Sp}}}\right), \label{eq:unsafe_set_matrix_spectrahedron} \\
    A_{\mathfrak{Sp}} & \isdef \matl A_{0, \mathfrak{Sp}} & A_{1, \mathfrak{Sp}} & A_{2, \mathfrak{Sp}} \matr. \nn
\end{align}
\normalsize
Note that the projection of a state vector $x_k$ onto the set $\overline{S}_{\mathfrak{Sp}}$ requires the solution of the convex, semidefinite optimization problem given by \eqref{eq:convex_projection}; in this work, this optimization problem is solved using the \texttt{MOSEK} solver interfaced with the \texttt{YALMIP} toolbox.
Next, it follows from the unsafe set \eqref{eq:unsafe_set_spectrahedron} that the safe set corresponding to the spectrahedron obstacle is given by
\small
\begin{align}
    S_{\mathfrak{Sp}} &= \left\{ p \in \BBR^2 \colon -(A_{0, \mathfrak{Sp}} + z_1 A_{1, \mathfrak{Sp}} + z_2 A_{2, \mathfrak{Sp}}) \nprec 0, \right. \nn \\
    & \quad \quad \left. \mbox{where } \matl z_1 \\ z_2  \matr = R_{\theta_{\mathfrak{Sp}}} (p - p_{\mathfrak{Sp}}) \right\},
\end{align}
\normalsize
since positive semidefiniteness does not hold if at least one of the eigenvalues of the associated real symmetric matrix is negative.
In order to derive a CBF for this indefinite inequality, a discrete-time version of the CBF constraint proposed in (19) from \cite{ong2025matrix} is formulated.
Using an indefinite matrix CBF was required for the implementation of the CBF associated with the spectrahedron, and we will explore this connection in more detail in future work.
Hence, the safe input constraint associated with the spectrahedron safe set $S_{\mathfrak{Sp}}$ with buffer ratio $\epsilon_{\mathfrak{Sp}} \in [0, 1)$ that replaces \eqref{eq:nonconvex_constraint} is given by
\begin{align}
    H_{\mathfrak{Sp}} & (f(x_k) + g(x_k)u) \nn \\ &\ge (1 + c_\perp) H_{\mathfrak{Sp}} (x_k) - (\gamma + c_\perp)  \ \lambda_{\mathfrak{Sp}, {\rm max}} \ I_{n_{\mathfrak{Sp}}}, \label{eq:nonconvex_constraint_spectrahedron}
\end{align}
where $\gamma \in [0, 1],$ $c_\perp \in [0, 1],$ and
\footnotesize
\begin{align*}
    &H_{\mathfrak{Sp}} \isdef  - A_{\mathfrak{Sp}} \left(\matl 1 \\ (1 - \epsilon_{\mathfrak{Sp}}) R_{\theta_{\mathfrak{Sp}}} (C_{\rm pos} x - p_{\mathfrak{Sp}}) \matr \otimes I_{ n_{\mathfrak{Sp}}}\right),\\
    &\lambda_{\mathfrak{Sp}, {\rm max}} \isdef {\rm max}({\rm eig}(H_{\mathfrak{Sp}} (x_k))).
\end{align*}
\normalsize
Note that, unlike the other obstacles, $\epsilon_{\mathfrak{Sp}}$ represents a ratio which increases the size of the unsafe set considered by the CBF the closer its value is to 1, since distance cannot be intuitively defined relative to any of the spectrahedron parameters.

\begin{figure}[!ht]
    \centering
    \includegraphics[width = 0.9\columnwidth]{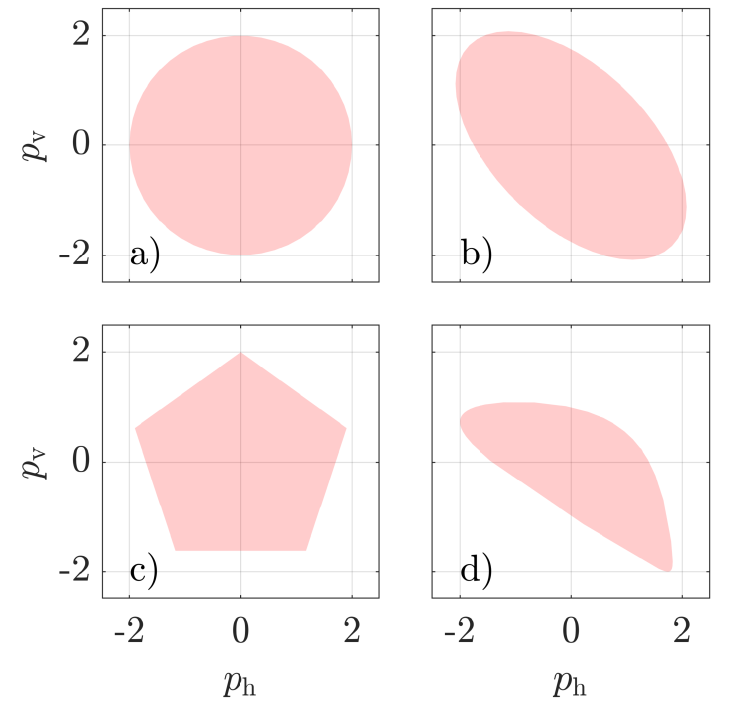}
    \caption{
    Typical obstacle geometries. 
    The figure shows a) circle (\ref{ssec:circle}), b) ellipse (\ref{ssec:ellipse}), c) convex polytope (\ref{ssec:polytope}), and d) spectrahedron (\ref{ssec:spectra}) shaped obstacles.
    %
    }
    \label{fig:ex_obstacles}
\end{figure}

\subsection{Discretized double integrators}\label{subsec:ex_double_integrator}

Robotic agents are frequently represented using double-integrator dynamics, which capture the relationship between acceleration, velocity, and position.
In this example, we consider the design of a safety-critical controller for an agent whose dynamics are governed by a double-integrator model.
In particular, consider two discretized double integrators
\begin{align}
    \matl p_{\rmh, k+1} \\ v_{\rmh, k+1} \matr &= A_\rmd \matl p_{\rmh, k} \\ v_{\rmh, k} \matr + B_\rmd u_{\rmh, k}, \label{eq:dynamics_example_1_h} \\
    \matl p_{\rmv, k+1} \\ v_{\rmv, k+1} \matr &= A_\rmd \matl p_{\rmv, k} \\ v_{\rmv, k} \matr + B_\rmd u_{\rmv, k}, \label{eq:dynamics_example_1_v},
\end{align}
where $p_{\rmh, k}, p_{\rmv, k} \in \BBR$ denote the translations in the horizontal and vertical directions, respectively, $v_{\rmh, k}, v_{\rmv, k} \in \BBR$ denote the velocities in the horizontal and vertical directions, respectively, $u_{\rmh, l}, u_{\rmv, k} \in \BBR$ are horizontal and vertical acceleration commands, respectively, and the discretized state transition matrices are 
\begin{align}
    A_\rmd \isdef \matl 1 & T_\rms \\ 0 & 1 \matr, \quad 
    B_\rmd \isdef \matl T_\rms^2/2 \\ T_\rms \matr.
    \label{eq:dynamics_example_1_AB}
\end{align}
Note that \eqref{eq:dynamics_example_1_h}, \eqref{eq:dynamics_example_1_v} can be written as
\begin{equation}
    x_{k+1} = A x_k + B u_k, \label{eq:dynamics_example_1}
\end{equation}
where $x_k \isdef \matl  p_{\rmh, k} & v_{\rmh, k} & p_{\rmv, k} & v_{\rmv, k}  \matr^\rmT,$ $u_k \isdef \matl u_{\rmh, k} & u_{\rmv, k} \matr^\rmT,$ $A \isdef {\rm diag} (A_\rmd, A_\rmd),$ and $B \isdef {\rm diag} (B_\rmd, B_\rmd).$
In this example, $C_{\rm pos} = \matl 1 & 0 & 0 & 0 \\ 0 & 0 & 1 & 0\matr,$ and we set $T_\rms = 0.01$ s.
Let $r_k \in \BBR^4$ be a reference state for $x_k.$
Hence, the objective of the nominal controller is to minimize $\sum_{k = 0}^{\infty} \Vert r_k - x_k \Vert.$

Consider a reference tracking nominal controller 
\begin{align*}
    u_{{\rm nom} , k} = f_{\rm lqr} (x_k, r_k),
\end{align*}
where $u_{{\rm nom}, k} \in \BBR$ is the requested control input and $f_{\rm lqr}$ encodes an LQR controller.
As mentioned at the beginning of this section, the details of the LQR controller implementation are omitted for brevity and since the main focus of the paper is the performance of the CBF.

The obstacles considered in this example are the following:
\begin{itemize}
    \item A circle obstacle $\mathfrak{C}$ with $p_\mathfrak{C} = \matl 18 & 16 \matr^\rmT$ m, and $r_\mathfrak{C} = 1.5$ m.
    \item An ellipse obstacle $\mathfrak{El}$ with $p_\mathfrak{El} = \matl 2.5 & 5 \matr^\rmT$ m, $v_{\mathfrak{El}} = \matl 1/\sqrt{2} & -1/\sqrt{2} \matr^\rmT,$ and 
    %
    %
    $\ell_{\mathfrak{El}, 1} = \sqrt{10}$ m, $\ell_{\mathfrak{El}, 2} = \sqrt{2},$
    \item A convex polytope obstacle $\mathfrak{P}$ given by a regular pentagon ($n_\mathfrak{P} = 5$) with center at $\matl 14 & 11 \matr^\rmT$ and a distance of 2 m from the center to each of its vertices.
    \item A spectrahedron obstacle $\mathfrak{Sp}$ with $p_{\mathfrak{Sp}} = \matl 7.5 & 7.5 \matr^\rmT$ m, $n_{\mathfrak{Sp}} = 3,$ $\theta_{\mathfrak{Sp}} = 0,$ and
    \footnotesize
    \begin{align*}
        A_{0, \mathfrak{Sp}} = 2 I_3, \quad & A_{1, \mathfrak{Sp}} = \matl  0 & 0.8 & 0 \\ 0.8 & 0 & 0.8 \\ 0 & 0.8 & 0 \matr,\\ & A_{2, \mathfrak{Sp}} = \matl  0 & 0 & 1.6 \\ 0 & 0 & 2.4 \\ 1.6 & 2.4 & 0 \matr.
    \end{align*}
    \normalsize
\end{itemize}

Two different CBF formulations are compared in the following example, the PDTE-MCBF formulation and a nonconvex CBF formulation:
\begin{itemize}
\item For the \textbf{PDTE-MCBF} formulation,
the control input $u_k$ is obtained by solving the constrained linear least-squares optimization problem
\begin{align}
    u_k &= \argmin_{\nu \in \BBR} \lVert \nu - u_{{\rm nom}, k} \rVert_2, \label{eq:ex1_pdte_mcbf} \\
    \mbox{s.t. } & H^\epsilon(A x_k + B \nu, x_k) \ggeq (1-\gamma) H^\epsilon(x_k,x_k), \nn
\end{align}
where $H^\epsilon$ is given by \eqref{eq:h_prime}, \eqref{eq:H_prime_set} with $Q = 4,$ each of the indices $i \in \{1, 2, 3, 4\}$ corresponding to each of the obstacles, $a_i, b_i^\epsilon$ given by \eqref{eq:a_i_b_i_pos}, and the projections $P_{\overline{S}_i}$ are calulated using \eqref{eq:convex_projection} with the corresponding $\overline{H}_i$ for each obstacle given by \eqref{eq:unsafe_set_matrix_circle}, \eqref{eq:unsafe_set_matrix_ellipse}, \eqref{eq:unsafe_set_matrix_polytope}, \eqref{eq:unsafe_set_matrix_spectrahedron}.
Algorithm \ref{alg:main_algorithm} is used to formulate and solve \eqref{eq:ex1_pdte_mcbf}.
This constrained linear least-squares optimization problem is solved at each iteration using the \texttt{lsqlin} Matlab solver with the \texttt{active-set} algorithm option.
This problem is warm-started by providing the nominal input $u_{{\rm nom}, k}$ as an initial solution.
Furthermore, the optimization problem formulations to calculate the projections corresponding to each obstacle are discussed in Subsection \ref{subsec:ex_obstacles}.
For this example, $\gamma = 0.2,$ and $\varepsilon = 0.4.$ \\
\item For the \textbf{nonconvex CBF} formulation,
the control input $u_k$ is obtained by solving the nonconvex, semidefinite optimization problem
\begin{align}
    u_k &= \argmin_{\nu \in \BBR} \lVert \nu - u_{{\rm nom}, k} \rVert_2, \\
    \mbox{s.t. } & \begin{matrix} h_{\mathfrak{C}} (A x_k + B \nu) \ge (1 - \gamma) h_{\mathfrak{C}} (x_k) \\ h_{\mathfrak{El}} (A x_k + B \nu) \ge (1 - \gamma) h_{\mathfrak{El}} (x_k) \\ \eqref{eq:nonconvex_constraint_polytope} \\ \eqref{eq:nonconvex_constraint_spectrahedron} \end{matrix},
\end{align}
with $h_{\mathfrak{C}}, h_{\mathfrak{El}}$ given by \eqref{eq:h_circle}, \eqref{eq:h_ellipse}, respectively, and $f(x_k) = A x_k, \ g(x_k) = B$ in \eqref{eq:nonconvex_constraint_polytope}, \eqref{eq:nonconvex_constraint_spectrahedron}.
This nonconvex, semidefinite optimization problem is solved at each iteration using the \texttt{bmibnb} solver from the \texttt{YALMIP} toolbox, interfaced with \texttt{MOSEK}.
This problem is warm-started by providing the nominal input $u_{{\rm nom}, k}$ as an initial solution.
For this example, $\gamma = 0.2,$ $\epsilon_{\mathfrak{C}} = \epsilon_{\mathfrak{El}} = \epsilon_{\mathfrak{P}} = 0.4,$ $\epsilon_{\mathfrak{Sp}} = 0.25,$ and $c_\perp = 0.05.$
\end{itemize}

For all simulations, we set $x_0 = 0,$ and $r_k \equiv \matl 18 & 0 & 20 & 0 \matr^\rmT.$
The simulations are run for all $k\in\{0, k_{\rm end}\},$ with $k_{\rm end} = 1200.$
The results of the simulations with the PDTE-MCBF and nonconvex CBF formulations are shown in Figures \ref{fig:ex_double_integrator_v1} and \ref{fig:ex_double_integrator_v2}, respectively.
The optimization problems run at each iteration returned feasible solutions in both simulations.
Figures \ref{fig:ex_double_integrator_v1} and \ref{fig:ex_double_integrator_v2} show the trajectory of the double integrators, the unsafe sets given by the obstacles, and the buffer sets, which correspond to the unsafe sets considered by the corresponding CBF formulations and facilitated by the buffer distances $\epsilon, \epsilon_{\mathfrak{C}}, \epsilon_{\mathfrak{El}}, \epsilon_{\mathfrak{P}}$ and the buffer ratio $\epsilon_{\mathfrak{Sp}}.$
The following can be concluded from both figures:
\begin{itemize}
\item The boundaries of the buffer sets resulting from the PDTE-MCBF formulation are more equidistant from the obstacle boundaries than the boundaries of the buffer sets resulting from the nonconvex CBF formulation;
this is better observed when comparing the buffer sets corresponding to the pentagon and spectahedron obstacles in both figures. 
\item The buffer sets resulting from the PDTE-MCBF formulation are similar to the Minkowski sum of the unsafe sets and a circle of radius $\epsilon$ m.
\item The tracking and obstacle avoidance performances of both CBF formulations are very similar.
\end{itemize}
Finally, Table \ref{tab:ex1_avg_runtime_per_iter} shows the average runtime per iteration for all CBF formulations for each of the simulations. In the case of PDTE-MCBF, the runtime without calculation of projected points refers to the runtime for solving \eqref{eq:ex1_pdte_mcbf} without calculating the projections $P_{\overline{S}_i}$, and the runtime with calculation of projected points includes the runtime for solving \eqref{eq:ex1_pdte_mcbf} considering the entirety of Algorithm \ref{alg:main_algorithm}.
The results in Table \ref{tab:ex1_avg_runtime_per_iter} show that the total runtime of the PDTE-MCBF formulation is more than an order of magnitude faster than the nonconvex CBF formulation.

\begin{figure}[!ht]
    \centering
    \includegraphics[width = \columnwidth]{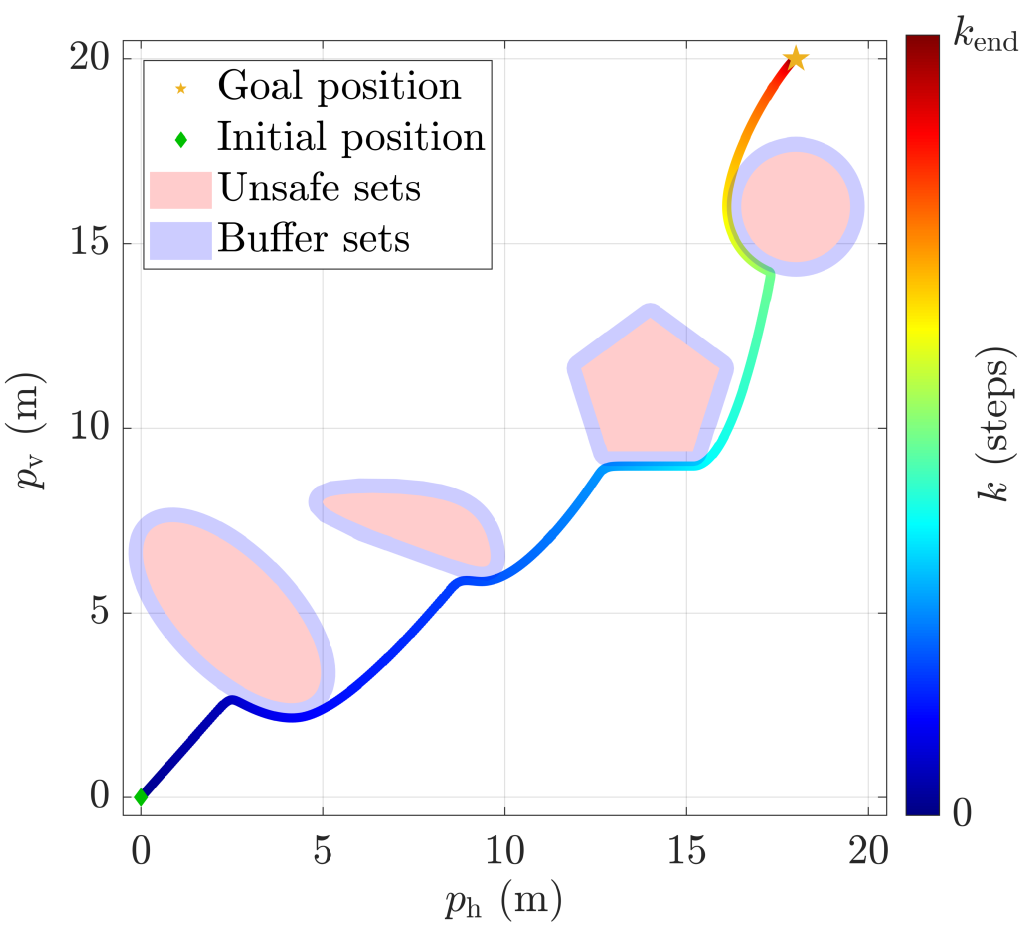}
    \caption{
    \textbf{Discretized double integrators with PDTE-MCBF formulation.}
    The figure shows the trajectory of an agent modeled by double integrators for all $k \in [0, k_{\rm end}]$ with $k_{\rm end} = 1200,$ the unsafe sets given by the obstacles, and the buffer sets, which are obtained by setting
    the buffer distance $\epsilon = 0.4.$
    }
    \label{fig:ex_double_integrator_v1}
\end{figure}

\begin{figure}[!ht]
    \centering
    \includegraphics[width = \columnwidth]{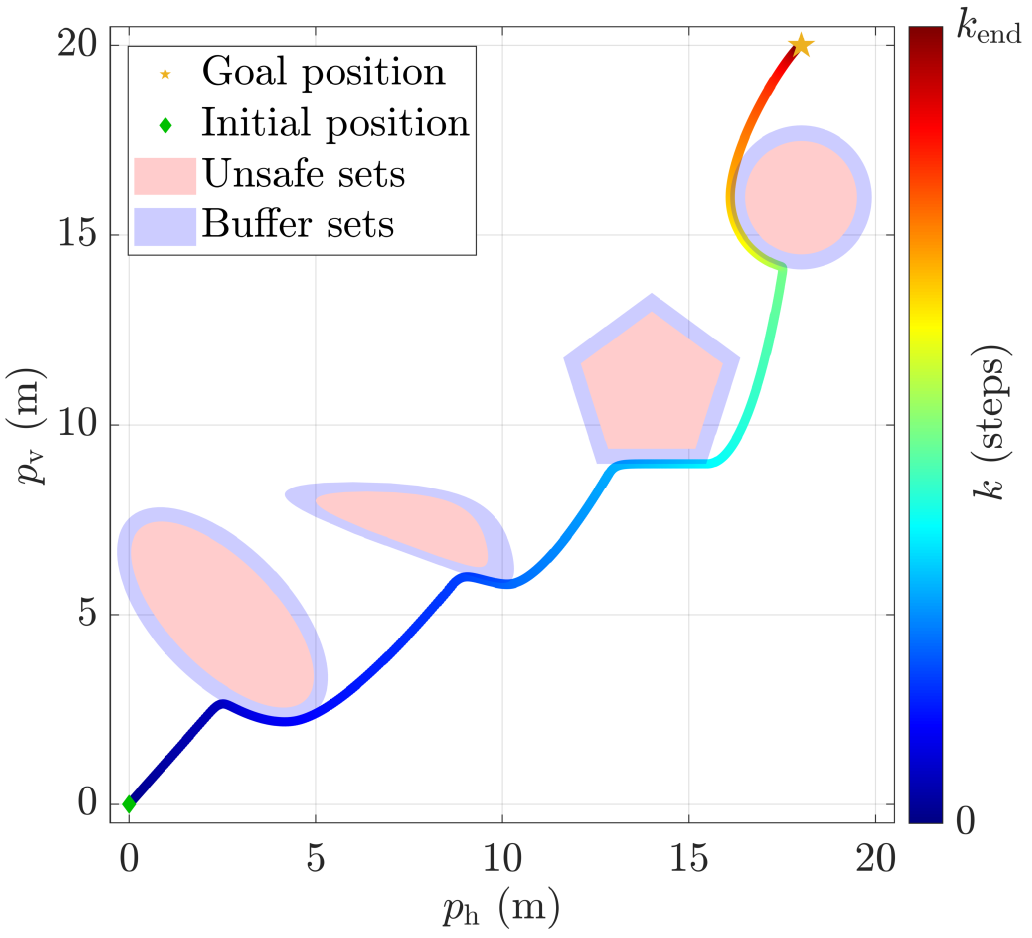}
    \caption{
    \textbf{Discretized double integrators with nonconvex CBF formulation.}
    This figure shows the trajectory of an agent modeled by double integrators for all $k \in [0, k_{\rm end}]$ with $k_{\rm end} = 1200,$ the unsafe sets given by the obstacles, and the buffer sets, which are obtained by setting 
    the buffer distances $\epsilon_{\mathfrak{C}} = \epsilon_{\mathfrak{El}} = \epsilon_{\mathfrak{P}} = 0.4$ and the buffer ratio $\epsilon_{\mathfrak{Sp}} = 0.25.$
    }
    \label{fig:ex_double_integrator_v2}
\end{figure}

\begin{table}[ht]
    \centering
    \renewcommand{\arraystretch}{1.2}
    \caption{Average runtime per iteration for all CBF formulations. The runtime without calculation of projected points refers to the runtime for solving \eqref{eq:ex1_pdte_mcbf} without calculating the projections $P_{\overline{S}_i}.$ The runtime with calculation of projected points includes the runtime for solving \eqref{eq:ex1_pdte_mcbf} considering the the entirety of Algorithm \ref{alg:main_algorithm}.}
    \resizebox{\columnwidth}{!}{%
    \begin{tabular}{|l|c|}
    \hline
    \rowcolor{cyan!70!black}
    \textcolor{white}{CBF formulation} 
    & 
    \textcolor{white}{$\begin{array}{c} \mbox{Average runtime} \\ \mbox{per iteration (s)}\end{array}$}
    \\
    \hhline{==}
        $\begin{array}{l} \mbox{PDTE-MCBF} \\ \mbox{(without calculation of projected points)}\end{array}$
        & 
        1.036 $\times 10^{-3}$
    \\
    \hline
        $\begin{array}{l} \mbox{PDTE-MCBF} \\ \mbox{(with calculation of projected points)}\end{array}$
        &
        6.002 $\times 10^{-2}$
    \\
    \hline
        \hspace{0.2em} Nonconvex CBF
        &
        2.543
    \\
    \hline
    \end{tabular}
    }
    \label{tab:ex1_avg_runtime_per_iter}
\end{table}

\subsection{Bicopter lateral flight}\label{subsec:ex_multicopter}

Consider the bicopter in the vertical plane shown in Figure \ref{fig:bicopter_diagram}, which consists of a rigid frame with two rotors that generate thrust along their respective axes. 
The bicopter has mass $m$, center of mass $c$, moment of inertia $J$ about $c$, and the distance between the rotors is $\ell_{\rm mc}$. 
Let $T_1, T_2$ denote the thrusts produced by the left and right rotors, respectively, as shown in Figure \ref{fig:bicopter_diagram}. 
Define the total thrust $T \isdef T_1 + T_2$ and the total moment $\tau \isdef (T_1 - T_2)/\ell_{\rm mc}$. 
Then, the dynamics of the bicopter in the vertical plane are then given by
\small
\begin{align}
    \dot{p}_\rmh &= v_\rmh, &
    \dot{v}_\rmh &= \frac{T}{m} \sin \theta,  \label{eq:bicopter_dyn_1} \\ 
    \dot{p}_\rmv &= v_\rmv, &
    \dot{v}_\rmv &= -\frac{T}{m} \cos \theta + g, \\
    \dot{\theta} &= \omega, &
    \dot{\omega} &= \frac{\tau}{J}, \label{eq:bicopter_dyn_6}
\end{align}
\normalsize
where $p_\rmh, p_\rmv \in \BBR$ are the horizontal and vertical positions of $c,$ respectively, $v_\rmh, v_\rmv \in \BBR$ are the horizontal and vertical velocities of $c,$ respectively, $\theta \in \BBR$ is the bicopter tilt, $\omega \in \BBR$ is the bicopter angular velocity, and $g$ is the acceleration due to gravity.
%

\begin{figure}[!ht]
    \centering
    \includegraphics[width = 0.6\columnwidth]{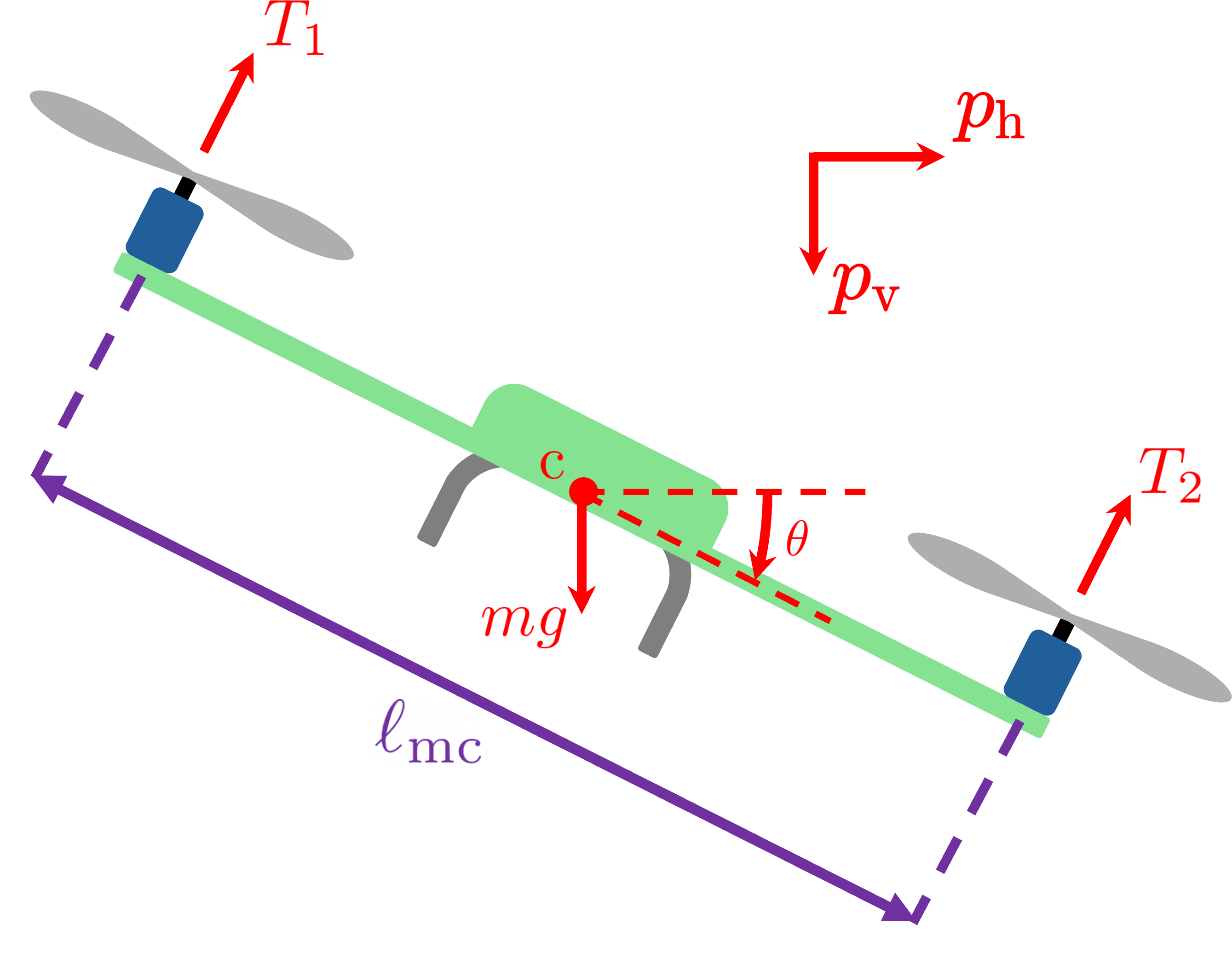}
    \caption{A bicopter system whose motion is constrained in a vertical plane.}
    \label{fig:bicopter_diagram}
\end{figure}

A discrete-time controller is implemented to control the continuous-time dynamics shown in \eqref{eq:bicopter_dyn_1}-\eqref{eq:bicopter_dyn_6}.
Hence, the states $p_\rmh, v_\rmh, p_\rmh, v_\rmh, \theta, \omega$ are sampled to produce the sampled states
\small
\begin{align*}
    p_{\rmh, k} &\isdef p_\rmh (k T_\rms), & v_{\rmh, k} &\isdef v_\rmh (k T_\rms), \\
    p_{\rmv, k} &\isdef p_\rmv (k T_\rms), & v_{\rmv, k} &\isdef v_\rmv (k T_\rms), \\
    \theta_k &\isdef \theta (k T_\rms), & \omega_k &\isdef \omega (k T_\rms),
\end{align*}
\normalsize
where $k \ge 0$ is the discrete-time step, and $T_\rms > 0$ is the sampling time.
The controller generates the total thrust $T_k \ge 0$ and the total torque $\tau_k \in \BBR.$
The continuous-time signals $T$ and $\tau$ applied to the bicopter are generated by applying a zero-order hold operation to $T_k$ and $\tau_k,$ that is, for all $k\ge0,$ 
\begin{equation}
    T(t) = T_k, \ \tau(t) = \tau_k, \mbox{ for all } t \in [kT_\rms, (k+1) T_\rms).
\end{equation}

The nominal controller is designed so that $p_{\rmh, k}$ and $p_{\rmv, k}$ follow reference signals $r_{\rmh, k}$ and $r_{\rmv, k},$ respectively, such that the objective of the nominal controller is to minimize $\sum_{k = 0}^{\infty} \left\Vert \matl r_{\rmh, k} & r_{\rmv, k} \matr^\rmT - \matl p_{\rmh, k} & p_{\rmv, k} \matr^\rmT \right\Vert.$
For this purpose, the inner-loop, outer-loop control architecture shown in Figure \ref{bicopter_cbf_blk_diag} is adopted.
An advantage of this architecture is that it allows the dynamics shown in \eqref{eq:bicopter_dyn_1}-\eqref{eq:bicopter_dyn_6} to be decoupled into linear systems, such that the resulting decoupled, discretized dynamics are given by
\begin{align}
    x_{\rmh, k+1} & = A_{\rm pos} x_{\rmh, k} + B_{\rm pos} u_{\rmh, k}, \label{eq:bicopter_dsg_1}\\
    x_{\rmv, k+1} & = A_{\rm pos}x_{\rmv, k} + B_{\rm pos} u_{\rmv, k}, \\
    x_{{\rm att}, k+1} & = A_{\rm att} x_{{\rm att}, k} + B_{\rm att} \tau_k, \label{eq:bicopter_dsg_3}
\end{align}
where $x_{\rmh, k} \isdef \matl p_{\rmh, k} & v_{\rmh, k} \matr^\rmT,$ $x_{\rmv, k} \isdef \matl p_{\rmv, k} & v_{\rmv, k} \matr^\rmT,$ $x_{{\rm att}, k} \isdef \matl \theta_k & \omega_k \matr^\rmT,$ $u_{\rmh, k}, u_{\rmv, k} \in \BBR$ are horizontal and vertical acceleration commands, respectively, and 
\begin{align*}
    A_{\rm pos} \isdef A_{\rm att} &\isdef \matl 1 & T_\rms \\ 0 & 1 \matr, \\
    B_{\rm pos} \isdef \matl T_\rms^2 / (2 m) \\ T_\rms / m \matr, & \quad B_{\rm att} \isdef \matl T_\rms^2 / (2 J) \\ T_\rms / J \matr.
\end{align*}
The dynamics shown in \eqref{eq:bicopter_dsg_1}--\eqref{eq:bicopter_dsg_3} only hold near hover conditions ($\theta \approx 0$ deg) and are used to design the outer-loop and inner-loop controllers

Let the outer-loop controller $G_{\rmc, {\rm ol}}$ be given by two LQR controllers for the horizontal and vertical states separately, such that
\begin{align}
    u_{{\rm nom}, \rmh , k} &= f_{{\rm lqr}, \rmi, \rmh} \left( \matl r_{\rmh, k} - p_{\rmh, k} \\ -v_{\rmh, k} \matr \right), \\
    u_{{\rm nom}, \rmv , k} &= f_{{\rm lqr}, \rmi, \rmv} \left( \matl r_{\rmv, k} - p_{\rmv, k} \\ -v_{\rmv, k} \matr \right),
\end{align}
where $f_{{\rm lqr}, \rmi, \rmh}, f_{{\rm lqr}, \rmi, \rmv}$ implement LQR controllers with integrator states and different sets of gains;
Algorithm 1 in \cite{predictive_CBF_2025} shows an implementation the LQR controllers mentioned above.
As mentioned at the beginning of this section, more details of the LQR controllers are omitted for brevity and since the main focus of the paper is the performance of the CBF.

Next, in this example, the function $f_{\rm cbf}$ is a safety filter that implements either the PDTE-MCBF formulation or the nonconvex CBF formulation.
The inputs of $f_{\rm cbf}$ are $u_{{\rm nom}, k} \isdef \matl u_{{\rm nom}, \rmh , k} & u_{{\rm nom}, \rmv , k} \matr^\rmT$ and $x_k \isdef \matl x_{\rmh, k}^\rmT & x_{\rmv, k}^\rmT \matr^\rmT,$ and the output of $f_{\rm cbf}$ is $u_k \isdef \matl u_{\rmh , k} & u_{\rmv , k} \matr^\rmT.$
In this example, $C_{\rm pos} = \matl 1 & 0 & 0 & 0 \\ 0 & 0 & 1 & 0\matr.$
The implementation details for both CBF formulations are given later.

The outer and inner loops are linked by a nonlinear mapping function $f_{\rm map}$ that can be used to obtain the thrust $T_k$ and a reference tilt value $\theta_{\rmr, k}$ from $u_k,$ such that
\begin{equation}
    \matl T_k \\ \theta_{\rmr, k} \matr = f_{\rm map} (u_k) = \matl \sqrt{ u_{\rmh, k}^2 + \left(mg -u_{\rmv, k}\right)^2 } \\ {\rm atan2}(u_{\rmh, k}, \ mg - u_{\rmv, k}) \matr. \label{eq:ex2_f_map}
\end{equation}
Then, the inner-loop controller $G_{\rmc, {\rm il}}$ is given by a LQR controller for the attitude states, such that
\begin{align}
    \tau_k = f_{{\rm lqr}, \rmi, {\rm att}} \left( \matl \theta_{\rmr, k} - \theta_k \\ -\omega_k \matr \right),
\end{align}
where $f_{{\rm lqr}, \rmi, {\rm att}}$ implements the LQR controller whose implementation is shown in Algorithm 1 in \cite{predictive_CBF_2025}.
As mentioned at the beginning of this section, more details of the LQR controllers are omitted for brevity and since the main focus of the paper is the performance of the CBF.

 \begin{figure} [h!]
    \centering
    \resizebox{\columnwidth}{!}{%
    \begin{tikzpicture}[>={stealth'}, line width = 0.25mm]

    \node [input, name=ref]{};
    \node [sum, fill=green!20, right =0.75cm of ref] (sum2) {};
    \node[draw = none] at (sum2.center) {$+$};
    \node [smallblock, fill=green!20, rounded corners, right = 0.4cm of sum2 , minimum height = 0.6cm , minimum width = 0.7cm] (controller) {$G_{\rmc, {\rm ol}}$};
    \node [smallblock, fill=green!20, rounded corners, right = 0.75cm of controller.east , minimum height = 0.6cm , minimum width = 0.7cm] (cbf) {$f_{\rm cbf}$};
    \node [smallblock, fill=green!20, rounded corners, right = 0.75cm of cbf.east , minimum height = 1cm, minimum width = 0.7cm] (map) {$f_{\rm map}$};
    \node [smallblock, fill=green!20, rounded corners, below right = 1.25cm and 1.25 cm of map.center , minimum height = 0.6cm, minimum width = 0.7cm] (controller_IL) {$G_{\rmc, {\rm il}}$};
    \node [sum, fill=green!20, left = 0.3cm of controller_IL.west] (sum3) {};
    \node[draw = none] at (sum3.center) {$+$};
    \draw[->] ([xshift = -0.8cm]sum2.west) -- node [above, xshift = -0.2cm] {\scriptsize $\matl r_{\rmh, k} \\ 0 \\ r_{\rmv, k} \\ 0 \matr$} (sum2.west);
    \draw[->] ([yshift = -1.5cm]sum2.south) -- node [xshift = -0.5cm, yshift = -0.2cm] {\scriptsize $\matl p_{\rmh, k} \\ v_{\rmh, k} \\ p_{\rmv, k} \\ v_{\rmv, k} \matr$} node [xshift = 0.25 cm, yshift = 0.5cm] {$-$} (sum2.south);
    \draw[->] ([yshift = -1.3cm]sum2.south) -| (cbf.south);
    \draw[->] (sum2.east) -- (controller.west);
    \draw[->] (controller.east) -- node [above] {\small $u_{\rmr, k}$} (cbf.west);
    \draw[->] (cbf.east) -- node [above] {\small $u_k$} (map.west);
    \draw[->] ([yshift = -0.25cm]map.east) -| node [near end, xshift = 0.4cm, yshift = 0.1cm]{\scriptsize $\matl \theta_{\rmr, k} \\ 0 \matr$}(sum3.north);
    \draw[->] ([xshift = -2em]sum3.west) -- node [above, xshift = -0.25cm]{\scriptsize $\matl \theta_k \\ \omega_k \matr$} node [below, xshift = 0.15cm, yshift = 0.05cm]{$-$} (sum3.west);
    \draw[->] (sum3.east) -- (controller_IL.west);
    \draw[->] ([yshift = 0.25cm]map.east) -- node[above, near end, xshift = 0.2cm]{\small$T_k$} ([xshift = 2.5cm, yshift = 0.25cm]map.east);
    \draw[->] (controller_IL.east) -- ([xshift = 0.2cm]controller_IL.east) |- node[above, near end, xshift = -0.05cm]{\small$\tau_k$} ([xshift = 2.5cm, yshift = -0.25cm]map.east);
    \end{tikzpicture}
    }  
    \caption{Inner-loop, outer-loop control architecture used for bicopter tracking control.}
    \label{bicopter_cbf_blk_diag}
\end{figure}
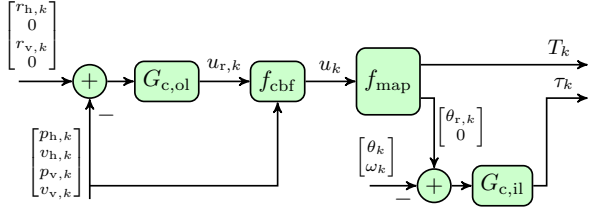

Next, safe set details are discussed.
In this example, the only obstacle considered is a circle obstacle  with $p_\mathfrak{C} = \matl 2 & -0.75 \matr^\rmT$ m, and $r_\mathfrak{C} = 1$ m,
which entails constraints on the position states, addressed by the previously introduced CBF formulations.
Aside from these constraints, velocity constraints and reference tilt constraints are considered to increase the likelihood that the safe inputs obtained by the CBFs do not cause a large deviation from hover conditions, which are assumed in the design of both the outer-loop and the inner-loop controllers.
Hence, for all steps $k\ge 0,$ the following velocity and reference tilt constraints are considered 
\begin{align}
    v_{\rmh, k} &\in [-v_{\rmh, {\rm max}}, v_{\rmh, {\rm max}}] \mbox{ m/s,} \label{eq:ex2_vh_constraint} \\
    v_{\rmv, k} &\in [-v_{\rmv, {\rm max}}, v_{\rmv, {\rm max}}] \mbox{ m/s,} \label{eq:ex2_vv_constraint} \\
    \theta_{\rmr, k} &\in [-\theta_{\rmr, {\rm max}}, \theta_{\rmr, {\rm max}}] \mbox{ deg,} \label{eq:ex2_tilt_r_constraint}
\end{align}
where $v_{\rmh, {\rm max}}, v_{\rmv, {\rm max}}, \theta_{\rmr, {\rm max}} > 0$ are the maximum absolute values of the horizontal and vertical velocities, and the reference tilt, respectively.
The CBF associated with the constraints \eqref{eq:ex2_vh_constraint}, \eqref{eq:ex2_vv_constraint} is given by
\begin{equation}
    h_{v}(x_k) = \matl 0 & 1 & 0 & 0 \\
                       0 & -1 & 0 & 0 \\
                       0 & 0 & 0 & 1 \\
                       0 & 0 & 0 & -1\matr x_k + 
                       \matl v_{\rmh, {\rm max}} \\ v_{\rmh, {\rm max}} \\ v_{\rmv, {\rm max}} \\ v_{\rmv, {\rm max}} \matr,
\end{equation}
and thus the safe input constraints associated with the velocity states are given by
\begin{equation}
    h_{v}(A x_k + B u_k) \ge (1-\gamma) h_{v}(x_k), \label{eq:ex2_vel_CBF_constraint}
\end{equation}
where $\gamma \in [0, 1],$ $A \isdef {\rm diag} (A_{\rm pos}, A_{\rm pos}),$ and $B \isdef {\rm diag} (B_{\rm pos}, B_{\rm pos}),$ which results in an affine, convex constraint.
Next, it follows from \eqref{eq:ex2_f_map} that safe input constraints associated with the reference tilt constraint \eqref{eq:ex2_tilt_r_constraint} is given by
\begin{equation}
    \matl 1 & -\rmt_{\theta_{\rmr, {\rm max}}} \\ -1 & -\rmt_{\theta_{\rmr, {\rm max}}} \matr u_k \ge - m g \ \rmt_{\theta_{\rmr, {\rm max}}} \matl 1 \\ 1 \matr, \label{eq:ex2_tilt_constraint}
\end{equation}
where $\rmt_{\theta_{\rmr, {\rm max}}} \isdef \tan \left(\tfrac{\pi}{180} \theta_{\rmr, {\rm max}}\right).$

Two different CBF formulations are compared in the following example, the PDTE-MCBF formulation and a nonconvex CBF formulation:
\begin{itemize}
\item For the \textbf{PDTE-MCBF} formulation,
the control input $u_k$ is obtained by solving the constrained linear least-squares optimization problem
\begin{align}
    u_k &= \argmin_{\nu \in \BBR} \lVert \nu - u_{{\rm nom}, k} \rVert_2, \label{eq:ex2_pdte_mcbf} \\
    \mbox{s.t. } & \begin{matrix} H^\epsilon(A x_k + B \nu, x_k) \ggeq (1-\gamma) H^\epsilon(x_k,x_k) \\ \eqref{eq:ex2_vel_CBF_constraint} \\ \eqref{eq:ex2_tilt_constraint} \end{matrix}, \nn
\end{align}
where $H^\epsilon$ is given by \eqref{eq:h_prime}, \eqref{eq:H_prime_set} with $Q = 1,$ and index $i = 1$ corresponding to the circle obstacle, $a_i, b_i^\epsilon$ given by \eqref{eq:a_i_b_i_pos}, and the projection $P_{\overline{S}_i}$ are calulated using \eqref{eq:convex_projection} with the corresponding $\overline{H}_i$ for the circle obstacle given by \eqref{eq:unsafe_set_matrix_circle}; in this case, the projection has a closed-loop solution given by \eqref{eq:proj_circle}.
Algorithm \ref{alg:main_algorithm} is used to formulate and solve \eqref{eq:ex2_pdte_mcbf}.
This constrained linear least-squares optimization problem is solved at each iteration using the \texttt{lsqlin} Matlab solver with the \texttt{active-set} algorithm option.
This problem is warm-started by providing the nominal input $u_{{\rm nom}, k}$ as an initial solution.
For this example, $\gamma = 0.1,$ and $\varepsilon = 0.2.$ \\
\item For the \textbf{nonconvex CBF} formulation,
the control input $u_k$ is obtained by solving the nonconvex, semidefinite optimization problem
\begin{align}
    u_k &= \argmin_{\nu \in \BBR} \lVert \nu - u_{{\rm nom}, k} \rVert_2, \label{eq:ex2_ncvx_CBF_prob} \\
    \mbox{s.t. } & \begin{matrix} h_{\mathfrak{C}} (A x_k + B \nu) \ge (1 - \gamma) h_{\mathfrak{C}} (x_k) \\ \eqref{eq:ex2_vel_CBF_constraint} \\ \eqref{eq:ex2_tilt_constraint} \end{matrix}, \nn
\end{align}
with $h_{\mathfrak{C}}$ given by \eqref{eq:h_circle}.
This nonconvex, nonlinear optimization problem is solved at each iteration using two different solvers in separate simulations.
The first considered solver is the \texttt{fmincon} Matlab solver with the \texttt{sqp} algorithm option and the maximum number of iterations set to $10^5.$
The second considered solver is the \texttt{bmibnb} solver from the \texttt{YALMIP} toolbox, interfaced with \texttt{MOSEK}.
In both instances, this problem is warm-started by providing the nominal input $u_{{\rm nom}, k}$ as an initial solution.
For this example, $\gamma = 0.1,$ and $\epsilon_{\mathfrak{C}} = 0.2.$
\end{itemize}

\begin{remark} \label{rem:solvers}
    The \texttt{lsqlin} and \texttt{fmincon} solvers can converge to unfeasible results that do not meet the optimization constraints and output these values.
    In contrast, the \texttt{YALMIP}-based solver does not output a numerical value if it determines that the problem is unfeasible. 
    In this case, the \texttt{YALMIP}-based solver is used a second time to solve \eqref{eq:ex2_ncvx_CBF_prob} with slack variables.
    If the \texttt{YALMIP}-based determines that the problem with slack variables is also unfeasible, the nominal control input is given as the output of the nonconvex, nonlinear optimization problem.
\end{remark}

For all simulations, $x (t) = 0, T (t) = mg, \tau (t) = 0$ for all $t \le 0,$ 
$u_{\rmh, k} = u_{\rmv, k} = \theta_{\rmr, k} = 0$ for all $k \le 0,$
$r_{\rmh, k} \equiv 18$ m, $r_{\rmv, k} \equiv 20$ m,
and $T_{\rm s} = 0.005$ s.
The simulations are run for all $t \in [0, t_{\rm end}]$ s, with $t_{\rm end} = 10$ s.
The Simulink\footnote{Special instructions for interfacing \texttt{YALMIP}-based solvers with Simulink are given in \href{https://yalmip.github.io/example/simulink/}{https://yalmip.github.io/example/simulink/}} environment is used for numerical simulation with the \texttt{ode45} solver to solve the bicopter continuous-time dynamics. 
The discrete-time dynamics corresponding to the controller and either of the CBF formulations are evaluated every $T_\rms$ seconds.

The results of the simulations with the PDTE-MCBF and nonconvex CBF formulations are shown in Figure \ref{fig:bicopter_ex2_closed_loop_response}.
Furthermore, the feasibility of the solutions that each solver converged to at each iteration step for the PDTE-MCBF and nonconvex CBF formulations is shown in Figure \ref{fig:bicopter_ex2_feasibility};
for the \texttt{YALMIP}-based solver, the problem solution is considered unfeasible if the first run without slack variables is considered unfeasible, as discussed in Remark \ref{rem:solvers}.
Finally, Table \ref{tab:ex2_bicopter_avg_runtime_per_iter} shows the average runtime per iteration for all CBF formulations for each of the solvers used in the simulations;
for the \texttt{YALMIP}-based solver, the runtime per iteration accounts for the runtime of the optimization problem with slack variables in the iterations where the first run without slack variables is considered unfeasible, as discussed in Remark \ref{rem:solvers}.

In all cases, the optimization problem becomes unfeasible as the bicopter enters the buffer set.
The safety filter is unable to prevent the bicopter from entering the buffer set since it only modifies the outer-loop controller outputs, which results in the inner-loop controller and the tilt angle dynamics becoming unmodeled input dynamics, which, in turn, can lead to safety violations, as discussed in \cite{seiler2021}.
However, while the nonconvex CBF formulation with the \texttt{fmincon} solver is unable to exit the buffer set and enters the unsafe set, the PDTE-MCBF formulation and the nonconvex CBF formulation with the \texttt{YALMIP}-based solver are able to recover from entering the buffer set and go back to the safe set before approaching the goal.
Furthermore, as shown in Table \ref{tab:ex2_bicopter_avg_runtime_per_iter}, the total runtime of the PDTE-MCBF formulation is faster than the nonconvex CBF formulation with the \texttt{fmincon} solver and is more than an order of magnitude faster than the nonconvex CBF formulation with the \texttt{YALMIP}-based solver.
A video illustrating these simulation results in the accompanying animation at
\href{https://youtu.be/eYDTqa__lE4}{https://youtu.be/eYDTqa\_\_lE4}.

\begin{figure}[!ht]
    \centering
    \includegraphics[width = 0.8\columnwidth]{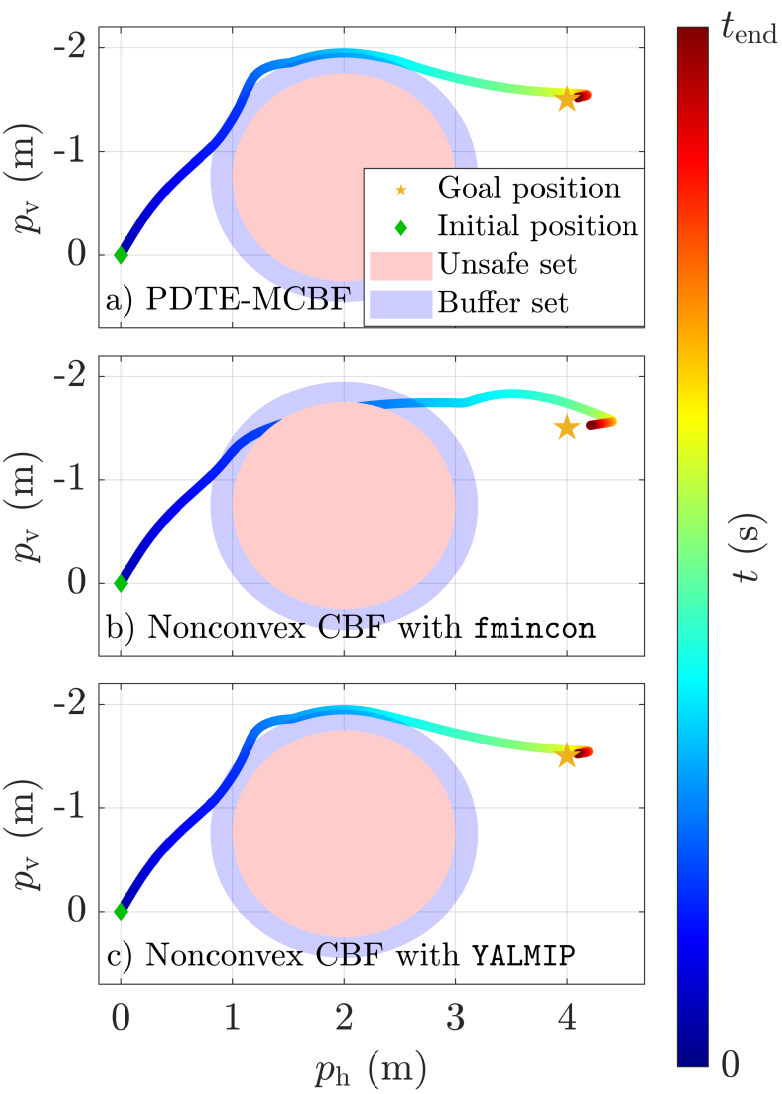}
    \caption{
    \textbf{PDTE-MCBF and nonconvex CBF closed-loop performance comparison in bicopter lateral flight system.}
    This figure shows the trajectory of the bicopters for all $t \in [0, t_{\rm end}]$ s, with $t_{\rm end} = 10$ s, the unsafe set given by the circle, and the buffer sets, which correspond to the unsafe sets considered by the CBF formulations and facilitated by the buffer distances $\epsilon = \epsilon_{\mathfrak{C}} = 0.2.$
    The plots show the results for a) the PDTE-MCBF formulation solved with \texttt{lsqlin}, b) the nonconvex CBF formulation solved with \texttt{fmincon}, and c) the nonconvex CBF formulation solved with the \texttt{YALMIP}-based solver.
    }
    \label{fig:bicopter_ex2_closed_loop_response}
\end{figure}

\begin{figure}[!ht]
    \centering
    \includegraphics[width = 0.85\columnwidth]{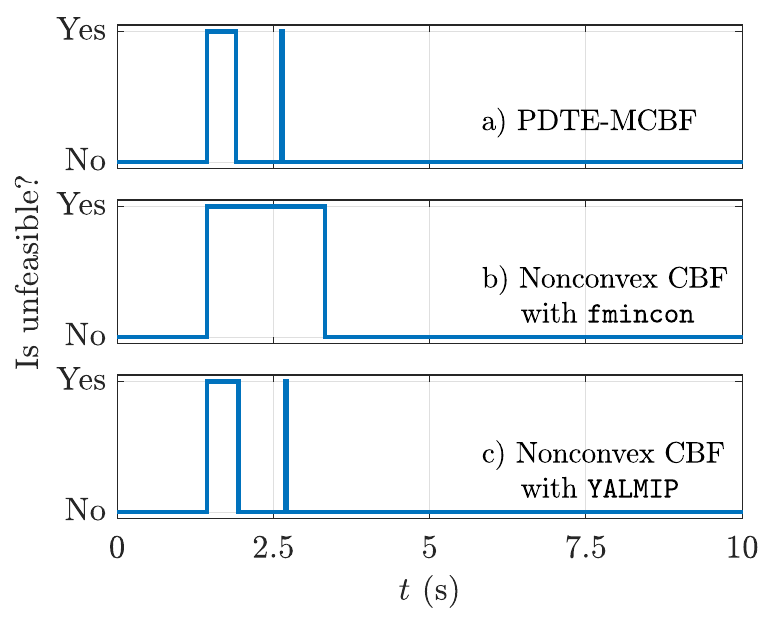}
    \caption{
    \textbf{PDTE-MCBF and nonconvex CBF closed-loop unfeasibility comparison in bicopter lateral flight system.}
    The iterations for which the implemented solvers converged to unfeasible solutions is shown in these plots.
    The plots show the results for a) the PDTE-MCBF formulation solved with \texttt{lsqlin}, b) the nonconvex CBF formulation solved with \texttt{fmincon}, and c) the nonconvex CBF formulation solved with the \texttt{YALMIP}-based solver.
    For the \texttt{YALMIP}-based solver, the problem solution is considered unfeasible if the first run without slack variables is considered unfeasible, as discussed in Remark \ref{rem:solvers}.
    }
    \label{fig:bicopter_ex2_feasibility}
\end{figure}

\begin{table}[ht]
    \centering
    \renewcommand{\arraystretch}{1.2}
    \caption{Average runtime per iteration for all CBF formulations with their respective solvers.}
    \resizebox{\columnwidth}{!}{%
    \begin{tabular}{|l|c|}
    \hline
    \rowcolor{cyan!70!black}
    \textcolor{white}{CBF formulation} 
    & 
    \textcolor{white}{$\begin{array}{c} \mbox{Average runtime} \\ \mbox{per iteration (s)}\end{array}$}
    \\
    \hhline{==}
        PDTE-MCBF with \texttt{lsqlin}
        & 
        5.251 $\times 10^{-6}$
    \\
    \hline
        Nonconvex CBF with \texttt{fmincon}
        &
        3.113 $\times 10^{-5}$
    \\
    \hline
        Nonconvex CBF with \texttt{YALMIP}
        &
        9.414 $\times 10^{-1}$
    \\
    \hline
    \end{tabular}
    }
    \label{tab:ex2_bicopter_avg_runtime_per_iter}
\end{table}

\section{Conclusion}
\label{sec:conclusions}

This paper introduced the notion of exponential discrete-time matrix control barrier functions, presented a novel method to compute safe control inputs using solely convex optimization called \emph{Projection-based Discrete-Time Exponential Matrix Control Barrier Function} (PDTE-MCBF), and demonstrated methods to handle disjunctive boolean constraints using discrete-time matrix control barrier functions. 
%
%
Numerical simulations showed the following advantages of the proposed method relative to nonconvex optimization methods:
\begin{itemize}
    \item The tracking and obstacle avoidance performances of both the PDTE-MCBF and the nonconvex CBF formulations are very similar.
    \item The boundaries of the buffer sets resulting from the PDTE-MCBF formulation are more equidistant from the obstacle boundaries than The boundaries of the buffer sets resulting from the nonconvex CBF formulation.
    \item The total runtime of the PDTE-MCBF formulation is more than an order of magnitude faster than the nonconvex CBF formulation.
\end{itemize}
Future work will investigate extensions of our results to multi-agent systems and systems with noise and uncertainty.

\section{Appendix: Matrix Convexity}
\label{sec:Appendix}

The following is an abbreviated overview of matrix convexity as defined in \cite[Sec. 3.6.2]{boyd2004convex}. More precisely, this definition of matrix convexity is simply \(\mathbb{S}^p\)-convexity for the proper cone \(\mathbb{S}^p\).

Let \(H : \R^n \to \mathbb{S}^p\) be a function with output in the space of symmetric real-valued \(p \times p\) matrices.

\begin{define}[{Matrix Convexity}]
    \label{def:matrix_convexity}
    A mapping \(H : \R^n \to \mathbb{S}^p\) is matrix convex if for all \(x,y \in \R^n\) and for all \(\theta \in [0,1]\), the following holds:
    \begin{align}
        H(\theta x + (1-\theta)y) \gleq \theta H(x) + (1 - \theta)H(y).
    \end{align}
    Equivalently, \(H\) is matrix convex if and only if \(x \mapsto y^\intercal H(x) y\) is convex in \(x\) for all \(y \in \R^p\). 
\end{define}

\begin{define}[Matrix Concavity]
    \label{def:matrix_concavity}
    A mapping \(H : \R^n \to \mathbb{S}^p\) is matrix concave if for all \(x,y \in \R^n\) and for all \(\theta \in [0,1]\), the following holds:
    \begin{align}
        H(\theta x + (1-\theta)y) \ggeq \theta H(x) + (1 - \theta)H(y).
    \end{align}
    Equivalently, \(H\) is matrix concave if and only if \(x \mapsto y^\intercal H(x) y\) is concave in \(x\) for all \(y \in \R^p\). 
\end{define}

\begin{lemma}
    \label{lem:scalar_ineq}
    Let \(A \in \mathbb{S}_+^p\). Then for all \(a,b \in \R\),
    \begin{align}
        a \leq b \iff aA \gleq bA.
    \end{align}
\end{lemma}

\begin{proof}
    Let \(U\Lambda U^T\) be the eigendecomposition of \(A\). We have
    \begin{align}
        (b-a)A = U\pth{(b-a)\Lambda}U^T.
    \end{align}
    Since \(A \in \mathbb{S}_+^p\), we have \(\lambda_i(A) \geq 0\). In addition, \(b-a \geq 0\) by assumption. It follows that \((b-a)A \ggeq 0\), which holds if and only if \(aA \gleq bA\).
\end{proof}

\begin{lemma}
    Let \(\{f_i : \R^n \to \R\}_{i=1}^q\) be convex functions, let \(\{A_i \in \mathbb{S}_+^p\}_{i=1}^q\) be positive semidefinite matrices, and let \(B \in \mathbb{S}^p\) be a symmetric matrix. Then the following function is matrix convex:
    \begin{align}
        H(x) \triangleq B + \sum_{i=1}^q A_i f_i(x).
    \end{align}
\end{lemma}

\begin{proof}
    Let \(\theta \in [0,1]\) and \(x,y \in \R^n\).
    By Lemma \ref{lem:scalar_ineq} and the convexity of each \(f_i\) it follows for all \(i=1,\ldots,q\) that
    \begin{align}
        A_i f_i(\theta x + (1-\theta)y) &\gleq A_i \pth{\theta f_i(x) + (1 - \theta)f_i(y)},\\
        &= \theta A_i f_i(x) + (1-\theta) A_i f_i(y).
    \end{align}
    Rearranging yields
    \small
    \begin{align}
        \theta A_i f_i(x) + (1-\theta) A_i f_i(y) - A_i f_i(\theta x + (1-\theta)y) \ggeq 0.
    \end{align}
    \normalsize
    Since the sum of PSD matrices is also PSD, we have
    \small
    \begin{align}
        \sum_{i=1}^q \theta A_i f_i(x) + (1-\theta) A_i f_i(y) - A_i f_i(\theta x + (1-\theta)y) \ggeq 0.
    \end{align}
    \normalsize
    Adding zero in the form \(0 = B - B\) yields
    \begin{align*}
        &\left(\sum_{i=1}^q \theta A_i f_i(x) + (1-\theta) A_i f_i(y) + B \right) \\
        & \hspace{2em}- \left(\sum_{i=1}^q A_i f_i(\theta x + (1-\theta)y) + B\right) \ggeq 0.
    \end{align*}
    Rearranging yields the desired inequality proving convexity:
    \begin{align*}
        &\sum_{i=1}^q A_i f_i(\theta x + (1-\theta)y) + B \gleq \\
        &\hspace{2em} \sum_{i=1}^q \theta A_i f_i(x) + (1-\theta) A_i f_i(y) + B
    \end{align*}
\end{proof}

\begin{lemma}
    \label{lem:matrix_convex_form}
    The mapping \(H : \R^n \to \mathbb{S}^p\) is matrix convex if \(H\) is in the form
    \begin{align}
        H(x) = B + \sum_{j=1}^J A_j f_j(x)
    \end{align}
    for some symmetric \(B \in \mathbb{S}^p\) and \(f_j : \R^n \to \R\), and at least one of the following conditions holds for each \(j \in 1,\ldots,J\):
    \begin{itemize}
        \item[1)] \(f_j\) is affine in \(x\) and \(A_j \in \mathbb{S}^p\),
        \item[2)] \(f_j\) is convex in \(x\) and \(A_j \ggeq 0\),
        \item[3)] \(f_j\) is concave in \(x\) and \(A_j \gleq 0\).
    \end{itemize}
\end{lemma}

\begin{proof}
    As per Definition \ref{def:matrix_convexity}, we proceed by proving that \(x \mapsto y^\intercal H(x) y\) is convex for all \(y \in \R^p\).
    Consider the function
    \begin{align}
        x &\mapsto y^\intercal \pth{B + \sum_{j=1}^J A_j f_j(x)} y, \\
        &\mapsto y^\intercal B y + \sum_{j=1}^J \pth{y^\intercal A_j y} f_j(x)
    \end{align}
    Case 1): If \(f_j\) is linear in \(x\), then \( y^\intercal B y + \sum_{j=1}^J \pth{y^\intercal A_j y} f_j(x)\) is an affine function of \(x\) which is both convex and concave.

    Case 2): Observe that \(y^\intercal A_j y \geq 0\) since \(A_j \ggeq 0\). This implies that \( y^\intercal B y + \sum_{j=1}^J \pth{y^\intercal A_j y} f_j(x)\) is the sum of convex functions each multiplied by a nonnegative scalar, added to a constant value. This results in a convex function.

    Case 3): Observe that \(y^\intercal A_j y \leq 0\) since \(A_j \gleq 0\). This implies that \( y^\intercal B y + \sum_{j=1}^J \pth{y^\intercal A_j y} f_j(x)\) is the sum of concave functions each multiplied by a nonpositive scalar, added to a constant value. This results in a convex function.
\end{proof}

\begin{lemma}
    The mapping \(H : \R^n \to \mathbb{S}^p\) is matrix concave if \(H\) is in the form
    \begin{align}
        H(x) = B + \sum_{j=1}^J A_j f_j(x)
    \end{align}
    for some symmetric \(B \in \mathbb{S}^p\) and at least one of the following conditions holds for each \(j \in 1,\ldots,J\):
    \begin{itemize}
        \item \(f_j\) is affine in \(x\) and \(A_j \in \mathbb{S}^p\),
        \item \(f_j\) is concave in \(x\) and \(A_j \ggeq 0\),
        \item \(f_j\) is convex in \(x\) and \(A_j \gleq 0\).
    \end{itemize}
\end{lemma}

\begin{proof}
    Follows from similar arguments as Lemma \ref{lem:matrix_convex_form}, using Definition \ref{def:matrix_concavity}.
\end{proof}

\printbibliography


\end{document}

%% file: PackagesCommands.tex
\usepackage{amsmath} 
\usepackage{amsfonts}
\usepackage{amsthm}

\newtheorem{theorem}{Theorem}

\theoremstyle{definition}

\newtheorem{problem}{Problem}[section]
\newtheorem{remark}{Remark}
\newtheorem{lemma}{Lemma}[section]

\usepackage{float} 
\usepackage[skip=0em,font=footnotesize]{caption}

\usepackage[margin = 1in]{geometry}

\usepackage{tikz}
\usetikzlibrary{shapes,arrows,calc,positioning}
\tikzstyle{bigblock} = [draw, fill=blue!20, rectangle, 
    minimum height=6em, minimum width=8em]
\tikzstyle{medblock} = [draw, fill=blue!20, rectangle, 
    minimum height=4em, minimum width=4em]    
\tikzstyle{mux} = [draw, fill=black!20, rectangle, 
    minimum height=5em, minimum width=0.1em]    
\tikzstyle{smallblock} = [draw, fill=blue!20, rectangle, 
    minimum height=2em, minimum width=3em]
    
\tikzstyle{data_block} = [draw, fill=green!20, rectangle, 
    minimum height=2em, minimum width=3em]
\tikzstyle{ops_block} = [draw, fill=blue!20, rectangle, 
    minimum height=2em, minimum width=3em]    
\tikzstyle{est_block} = [draw, fill=red!20, rectangle, 
    minimum height=2em, minimum width=3em]    
    
\tikzstyle{sum} = [draw, fill=blue!20, circle, node distance=1cm,minimum height=0.5cm]
\tikzstyle{signal} = [coordinate]
\tikzstyle{pinstyle} = [pin edge={to-,thin,black}]
\tikzstyle{block} = [draw, fill=blue!20, rectangle, 
    minimum height=3em, minimum width=9em]
\tikzstyle{blockS} = [draw, fill=blue!20, rectangle, 
    minimum height=3em, minimum width=4em]    
\tikzstyle{input} = [coordinate]
\tikzstyle{output} = [coordinate]
\usetikzlibrary{matrix}

\usetikzlibrary{positioning}
\usetikzlibrary{math}

\usepackage{hyperref}
\usepackage{xcolor}
\hypersetup{
    colorlinks,
    linkcolor={blue!100!black},
    citecolor={blue!50!black},
    urlcolor={blue!80!black}
}

\newcommand{\bc}{\begin{center}}
\newcommand{\ec}{\end{center}}
\newcommand{\benum}{\begin{enumerate}}
\newcommand{\eenum}{\end{enumerate}}
\newcommand{\nn}{\nonumber}
\newcommand{\matl}{\left[ \begin{array}}
\newcommand{\matr}{\end{array} \right]}

\renewcommand{\matl}{\begin{bmatrix}}
\renewcommand{\matr}{\end{bmatrix}}

\newcommand{\matls}{\left[ \begin{smallmatrix}}
\newcommand{\matrs}{\end{smallmatrix} \right]}
\newcommand{\isdef}{\stackrel{\triangle}{=}}

\newcommand{\rmT}{{\rm T}}

\newcommand{\rmc}{{\rm c}}
\newcommand{\rmd}{{\rm d}}

\newcommand{\rmh}{{\rm h}}
\newcommand{\rmi}{{\rm i}}

\newcommand{\rmr}{{\rm r}}
\newcommand{\rms}{{\rm s}}
\newcommand{\rmt}{{\rm t}}

\newcommand{\rmv}{{\rm v}}

\newcommand{\BBR}{{\mathbb R}}

\let\labelindent\relax
\usepackage{enumitem,amssymb}
\newlist{todolist}{itemize}{2}
\setlist[todolist]{label=$\square$}
\usepackage{pifont}

%% file: JUcmds.tex
\newtheorem{define}{Definition}
\newtheorem{assume}{Assumption}

\newcommand{\R}{\mathbb R}

\newcommand{\bmx}[1]{\begin{bmatrix}#1\end{bmatrix}} 
\newcommand{\bkt}[1]{\left[#1\right]} 
\newcommand{\pth}[1]{\left(#1\right)} 
\newcommand{\brc}[1]{\left \{#1\right \}} 
\newcommand{\nrm}[1]{\left \lVert#1\right \rVert} 

\newcommand{\bmxs}[1]{\begin{bsmallmatrix}#1\end{bsmallmatrix}}

\newcommand{\gleq}{\preceq} 
\newcommand{\ggeq}{\succeq} 

\DeclarePairedDelimiter{\ceil}{\lceil}{\rceil}
\DeclarePairedDelimiter{\floor}{\lfloor}{\rfloor}
\DeclarePairedDelimiter{\abs}{\lvert}{\rvert}

\makeatletter
\let\oldceil\ceil
\def\ceil{\@ifstar{\oldceil}{\oldceil*}}
\makeatother
\makeatletter
\let\oldfloor\floor
\def\floor{\@ifstar{\oldfloor}{\oldfloor*}}
\makeatother
\makeatletter
\let\oldnorm\norm
\def\norm{\@ifstar{\oldnorm}{\oldnorm*}}
\makeatother
\makeatletter
\let\oldabs\abs
\def\abs{\@ifstar{\oldabs}{\oldabs*}}
\makeatother